\pdfoutput = 1
\documentclass[11pt]{article}
\usepackage[utf8]{inputenc}
\usepackage{fullpage}
\usepackage{amssymb}
\usepackage{amsmath}
\usepackage{amsthm}
\usepackage[dvipsnames]{xcolor}
\usepackage{booktabs}

\newcommand{\abs}[1]{\left|#1\right|}
\newcommand{\set}[1]{\left\{#1\right\}}

\newcommand{\paren}[1]{\left(#1\right)}
\DeclareMathOperator{\E}{\mathbb{E}}
\newcommand{\cG}{\mathcal{G}}

\newcommand{\cD}{\mathcal{D}}

\newcommand{\cB}{\mathcal{B}}
\newcommand{\cA}{\mathcal{A}}

\newcommand{\cU}{\mathcal{U}}

\newcommand{\cS}{\mathcal{S}}
\newcommand{\sC}{{\mathcal C}}
\newcommand{\polylog}{\mathrm{polylog}}
\newcommand{\poly}{\mathrm{poly}}
\newcommand{\bR}{{\mathbb R}}
\newcommand{\bZ}{{\mathbb Z}}

\newcommand{\sps}[1]{^{(#1)}}
\newcommand{\one}{\mathbf{1}}
\newcommand{\st}{\textbf{ s.t. }}
\newcommand{\supp}{{\mathsf{supp}}}
\newcommand{\numcap}{{\mathsf{cap}}}
\newcommand{\lcap}{{\mathsf{Lcap}}}

\newcommand{\sign}{{\mathsf{sign}}}
\newcommand{\dkl}{{\mathsf{D_{KL}}}}
\newcommand{\ber}{{\mathsf{Ber}}}
\newcommand{\var}{{\mathsf{Var}}}

\newcommand{\acc}{\textup{``accept''}}
\newcommand{\rej}{\textup{``reject''}}

\usepackage[linktocpage,pagebackref]{hyperref}

\usepackage{cleveref}

\newtheorem{theorem}{Theorem}[section]
\newtheorem{definition}{Definition}[section]
\newtheorem{claim}[theorem]{Claim}
\newtheorem{lemma}[theorem]{Lemma}
\newtheorem{corollary}[theorem]{Corollary}
\newtheorem{remark}{Remark}[section]

\title{Generative Models of Huge Objects}
\date{}
\author{Lunjia Hu\thanks{Stanford University. Supported by the Simons Foundation Collaboration on the Theory of Algorithmic Fairness, Omer Reingold’s NSF Award IIS-1908774, and Moses Charikar's Simons Investigators award.} \and 
Inbal Livni-Navon\thanks{Stanford University. Supported by the Simons Foundation Collaboration on the Theory of Algorithmic Fairness, the Sloan Foundation Grant 2020-13941, and the Zuckerman STEM Leadership Program.} \and 
Omer Reingold\thanks{Stanford University. Supported by the Simons Foundation Collaboration on the Theory of Algorithmic Fairness and the Simons Foundation Investigators award 689988.}
}

\allowdisplaybreaks
\begin{document}
\maketitle
\thispagestyle{empty}
\begin{abstract}

This work initiates the systematic study of explicit distributions that are indistinguishable from a \emph{single} exponential-size combinatorial object. In this we extend the work of Goldreich, Goldwasser and Nussboim (SICOMP 2010) that focused on the implementation of huge objects that are indistinguishable from the uniform distribution, satisfying some global properties (which they coined truthfulness). Indistinguishability from a single object is motivated by the study of generative models in learning theory and regularity lemmas in graph theory. Problems that are well understood in the setting of pseudorandomness present significant challenges and at times are impossible when considering generative models of huge objects. 

We demonstrate the versatility of this study by providing a learning algorithm for huge indistinguishable objects in several natural settings including: dense functions and graphs with a truthfulness requirement on the number of ones in the function or edges in the graphs, and a version of the weak regularity lemma for {\em sparse} graphs that satisfy some global properties. These and other results generalize basic pseudorandom objects as well as notions introduced in algorithmic fairness.  The results rely on notions and techniques from a variety of areas including learning theory, complexity theory, cryptography, and game theory. 
\end{abstract}

\newpage
\tableofcontents
\newpage

\section{Introduction}

A pseudorandom distribution is indistinguishable from the uniform distribution to a set of computationally bounded distinguishers. Pseudorandomness is a cornerstone of many areas of computer science and mathematics. The variability of pseudorandom distributions stems from the different objects they can generate (bit strings, functions, permutations and more) and the different computational bounds that can be imposed on the distinguishers. In the area of cryptography, it is typical to consider powerful distinguishers that are at least polynomial time, giving rise to central notions such as pseudorandom generators~\cite{BlumMi84,Yao82}, pseudorandom functions~\cite{GoldreichGoMi86} and pseudorandom permutations~\cite{LubyR88}. More limited distinguishers give rise to other fundamental notions such as $k$-wise independent hashing and $\varepsilon$-biased distributions~(cf.~\cite{CarterWe81,NaorNa93}).  In the area of explicit combinatorial constructions, we typically try to emulate the uniform distribution by a single object, rather than with a distribution. A primer example is the fundamental notion of expander graphs (see \cite{HooryLiWi06,TCS-010} for surveys), with its multiple variants (including various notions of randomness extractors). These are graphs that are indistinguishable from a uniformly selected graph to a limited set of distinguishers (such as distinguishers that check if a random edge crosses a given cut). 

In these classic areas of pseudorandomness, a distribution, or even a single object, is constructed to emulate a distribution (typically, the uniform distribution). In this paper we ask for a distribution to emulate {\em a single object} (\Cref{table:1}). This reversal may seem absurd from the perspective of pseudorandomness but makes perfect sense from the perspective of generative models. An early exposure of the TOC community to generative models was with respect to the World Wide Web. These were models that produce distributions of graphs that imitate some properties of the Web, such as power law on the degrees of nodes (see~\cite{Mit03} for a survey). At any given point, the web is a single graph, but it is also a very large graph that does not have a simple description. Generative models gave a useful way to analyze, or estimate through experimentation, the expected performance of protocols on the Web. 

Other well studied generative models are the stochastic block model~\cite{HOLLAND1983109} and the more elaborate mixed membership stochastic block model~\cite{Airoldietal}. Consider a graph representing some connections between individuals, such as the connectivity of the social network. The stochastic block model partitions the vertices into disjoint communities and for every two communities assigns a probability of connection. This model represents a distribution over graphs where for each two vertices, an edge is placed independently with the probability assigned to the pair of communities of its end points. (In the mixed-membership model, each vertex is assigned a distribution over communities.) These models help identify useful substructures  within a social structure such as sub-communities or different social roles. But given a single social network, $B^*$, what is an appropriate model to capture it? After all, a model describes a distribution over networks rather than the single network we are trying to explain. A prevalent approach is to aim at the maximum likelihood model. Out of all models, the probability of sampling $B^*$ is maximized under the maximum likelihood model. Heuristics for estimating the maximum likelihood model have been playing a major roles in the generative-model literature and in its application in practice.  It should be noted that the probability that the model would produce $B^*$ is often very small. In this light, the meaningfulness of a maximum-likelihood models may be debated and may depend on a particular setting. 

From the perspective of indistinguishability, it may be more natural to seek a model that produce a distribution that is indistinguishable from $B^*$ to a meaningful set of distinguishers. For example, in the case of the stochastic block model, natural distinguishers are defined by two sets of vertices $U$ and $V$ and ask what is the probability that a random edge in the graph crosses from $U$ to $V$. A stochastic block model that fool all such distinguishers is exactly what is given by the Frieze-Kannan regularity lemma (also known as the weak regularity lemma)~\cite{Frieze}. The indistinguishability perspective on generative models and known connections between learning and pseudrandomness, which we will discuss shortly, are both a motivation as well as the starting point of this work. 

\begin{table}
\centering
\begin{tabular}{lll} 
 \toprule
& What do we imitate? & What do we construct?  \\ 
 \midrule
Pseudorandomness
 & distribution of objects & distribution of objects  \\ 
Explicit construction (e.g.\ expander graphs) & distribution of objects & single object  \\
 Our setup: generative models & single object & distribution of objects \\
  \bottomrule
\end{tabular}
\caption{Comparison between problem settings}
\label{table:1}
\end{table}

\subsection{Overall Goal: Indistinguishable Generative Models of Huge Objects} 
\label{sec:overall-goal}
In many of the applications of generative models, such as modeling the Web or a social network, the objects being modeled are huge. In this paper, we aim at a {\em systematic theory of efficiently learning and implementing huge generative models}. Our models will generate a distribution of objects satisfying some global properties that are indistinguishable from a fixed combinatorial object. Such a theory presents non-trivial challenges that do not manifest themselves neither when generating huge pseudorandom objects, nor in generative models of polynomial-size objects. 

Concretely, we assume that we have access to an object $B^*$ with exponential size. For example, $B^*$ could be a function with an exponentially large domain, or a graph with exponentially many vertices and edges. We are most interested in the case where the object $B^*$ is too large to read or process as a whole, and we have to access it by sampling: for example, when $B^*$ represents a function $f:X\to \{0,1\}$ with $|X|$ being exponentially large, we may access $B^*$ by asking for random pairs $(x,f(x))\in X\times \{0,1\}$ (sample access) or random inputs $x\in X$ conditioned on $f(x) = 1$ (support access). Given access to the huge object $B^*$, our goal is to create a generative model $M$ for $B^*$. Here, our model $M$ represents a distribution over objects, and we want to ensure that this distribution is indistinguishable from $B^*$ to all distinguishers $D$ in a class $\cD$. Specifically, if we use $D^B\in \{\acc,\rej\}$ to denote the output of distinguisher $D$ given sample/support access to object $B$, our indistinguishability requirement is that for every $D\in \cD$,
\[
|\Pr[D^{B^*}= \acc] - \E_{B\sim M}[\Pr[D^{B}= \acc]]|\le \varepsilon.
\]
We aim for building an \emph{efficient} learner $L$ that can output a model $M$ satisfying the indistinguishability requirement above when given sample/support access to the true object $B^*$. 
When $B^*$ is exponentially large (which is the case we are interested in), the output model $M$ also needs to generate exponentially large objects, and thus we cannot expect an efficient learner to directly output $M$.
Instead, we want our learner to output an \emph{efficient implementation} of $M$, which, roughly speaking, is a randomized algorithm that can efficiently provide sample/support access to objects drawn from $M$ (\Cref{def:ordinary-imp,def:random-imp}).

Our goal of learning a generative model $M$ indistinguishable from the true object $B^*$ is analogous to the problem addressed by Goldreich, Goldwasser and Nussboim \cite{goldreich2010} in the area of pseudorandomness.
They study the problem of efficiently implementing a distribution of huge objects, satisfying some global properties, that are indistinguishable from the uniform distribution of such objects. 
Follow-up works of \cite{goldreich2010} such as \cite{NaorNT05,NaorN07} study efficient implementations that are indistinguishable from certain distributions of huge random graphs.
All these works aim to achieve indistinguishability from a known distribution of objects, whereas in our problem of learning a generative model, we assume that
the true object $B^*$ is initially \emph{unknown}, and to collect information about $B^*$, we additionally need a learner $L$ that can use sample/support access to $B^*$ to \emph{efficiently} construct an implementation of an indistinguishable model.

Beyond indistinguishability, we also aim to achieve the notion of \emph{truthfulness} introduced in \cite{goldreich2010}. To demonstrate this notion, consider pseudorandom permutations $f_s:\{0,1\}^n\mapsto \{0,1\}^n$~\cite{LubyR88}. A distribution of permutations is pseudorandom if it is indistinguishable from the uniform distribution of permutations. It should be noted that a pseudorandom $\{0,1\}^n\mapsto \{0,1\}^n$ function~\cite{GoldreichGoMi86} is also indistinguishable from a random permutation over $\{0,1\}^n$ (as long as the number of queries are sufficiently smaller than $2^{n/2}$). Nevertheless, insisting that the pseudorandom objects satisfy the global condition of being a permutation is critical in the applications of pseudorandom permutations. This motivates the distinction of \cite{goldreich2010} between indistinguishability (that the pseudorandom objects are indistinguishable from a uniform object to a class of distinguishers) and {\em truthfulness} which is a global property that needs to hold exactly or approximately in a statistical sense. In our setup, a generative model $M$ is truthful if every object $B$ drawn from the distribution represented by $M$ satisfies a certain global property. For example, when the true object $B^*$ is a function $f^*:X\to \{0,1\}$ with support size $|\{x\in X:f^*(x) = 1\}|$ being $k$, a truthful requirement on a generative model $M$ for $B^*$ may restrict $M$ to always generate functions with support size $k$.

The study of implementing huge pseudorandom objects~\cite{10.1007/s00145-001-0008-5,goldreich2010,NaorNT05,NaorN07} has pseudorandom functions and permutations as vital building blocks.
Besides these building blocks, our techniques for learning generative models of huge objects also come from connections to the regularity lemma and especially the work of Trevisan, Tulsiani and Vadhan~\cite{TrevisanTV09}. In \cite{TrevisanTV09}, they construct an efficiently-implementable function $f:X\mapsto [0,1]$ which is indistinguishable from some $f^*:X\mapsto [0,1]$ to a family of distinguishers represented by functions $g:X\mapsto [0,1]$. Indistinguishability here means that $|\E[f(x)g(x)]-\E[f^*(x)g(x)]|$ is smaller than some error parameter $\varepsilon$.
After \cite{TrevisanTV09}, the problem of creating indistinguishable functions and its applications to cryptography are further studied in \cite{vadhan2013uniform,MR3183555,MR3591815,skorski2016subgradient,MR3655523,MR3794838}.
These works assume that the true function $f^*$ is known and they do not explicitly deal with the problem of learning $f^*$, but the corresponding learning task has been studied in the algorithmic fairness literature through the notions of multi-accuracy, multi-calibration, and outcome indistinguishability \cite{hebert2018multicalibration,kim2019multiaccuracy,dwork2021outcome,pmlr-v178-gopalan22a,dwork2023new}.
When applying techniques from these works to solve some problems in our setting, we need to deal with additional challenges such as the truthfulness requirement that we want our generative model to satisfy.

\subsection{Our Results}
The main conceptual contribution of this paper is in suggesting a new frontier for the study of indistinguishability, which is highly motivated and technically challenging. 
As we introduce in \Cref{sec:overall-goal},
the notion of indistinguishability from a single huge object combines at its core the areas of learning theory and pseudorandomness which, as recent research uncovered, have deep connections, providing a way to describe and address a rich landscape of natural problems. 
Below we summarize the main problems we address in this new framework.

\subsubsection*{Truthful Learning That Preserves Support Size}
Suppose we have sample access to a function $f^*:\{0,1\}^n \to \{0,1\}$ and we want to build an indistinguishable generative model for $f^*$. Here sample access allows us to observe pairs $(x,f^*(x))$ with $x$ drawn uniformly at random from $\{0,1\}^n$, and accordingly, we assume that every distinguisher also decides to accept or reject based on such a random pair $(x,f(x))$ from a function $f$ that may or may not be the true $f^*$. This task of learning a generative model for a binary function is closely related to the task of \emph{no-access outcome indistinguishability} studied in \cite{dwork2021outcome}, and it has been observed that the task can be reduced to multi-accuracy. Indeed, assuming that the distinguishers have bounded complexity and can be learned efficiently, using previous algorithms in \cite{hebert2018multicalibration,kim2019multiaccuracy,dwork2021outcome}, we can design an efficient learner that constructs a generative model indistinguishable from $f^*$ (\Cref{thm:sample-function}). The model constructed this way is specified using a predictor $p:\{0,1\}^n\to [0,1]$, and the model represents the distribution of functions $f:\{0,1\}^n\to \{0,1\}$ where the function value $f(x)$ is distributed independently for every $x\in \{0,1\}^n$ according to the Bernoulli distribution $\ber(p(x))$ with mean $p(x)$.

Learning generative models for binary functions becomes a more challenging task when we additionally enforce truthfulness requirements. A natural choice of truthfulness requirement is to preserve the support size of the function. Assuming that we know the support size $|\{x\in \{0,1\}^n:f^*(x) = 1\}|$ of $f^*$ is $k$, we would like our model to only generate functions that also have support size $k$. We show how to build an efficient learner that can output such a truthful model which is also indistinguishable from the true function $f^*$:

\begin{theorem}[Informal statement of \Cref{thm:fixed-weight}]
Let $\cB$ be the class of sample-access objects induced by binary functions $f:\{0,1\}^n\to\{0,1\}$ satisfying $|\supp(f)| = k$, where  $\supp(f):=\{x\in \{0,1\}^n:f(x) = 1\}$. Let $\cD$ be a class of distinguishers that is efficiently learnable. There exists an efficient $(\varepsilon,\delta)$-learner $L$ for $\cB$ w.r.t.\ $\cD$ and the learner always outputs an efficient implementation of a model $M$ that is truthful w.r.t.\ $\cB$.
\end{theorem}

Note that the truthfulness requirement on the support size cannot be enforced simply using computationally-bounded distinguishers, because computing the support size of a function $f$ exactly requires reading the values $f(x)$ for all the exponentially many inputs $x\in \{0,1\}^n$. Also, this truthfulness requirement cannot be satisfied directly by a generative model specified by a predictor $p$, where the function value $f(x)$ is distributed according to $\ber(p(x))$ independently of the function values $f(x')$ of other individuals $x'\ne x$. To enforce a fixed support size, the function values of different individuals must coordinate in a global manner, requiring us to use new techniques. We create a binary tree with leaves corresponding to the function domain $\{0,1\}^n$, and following ideas in previous work such as \cite{goldreich2010}, we assign support size budgets from the root to the leaves. However, the true function $f^*$ is a single unknown object and is very different from the uniform distribution considered in \cite{goldreich2010}, so there are no closed-from distributions (such as the binomial distributions used in \cite{goldreich2010}) that can guide us to distribute the support size budget from a node two its two children. Instead, we need to \emph{estimate} how the budget should be divided, and this leads to accumulated error towards the leaves and forces us to stop before reaching the leaves. To efficiently propagate the budgets to the leaves, we solve a zero-sum game where player $C$ chooses the budgets for the leaves and player $D$ distinguishes them from the target. We show that if player $D$ uses the multiplicative weights algorithm to minimize regret, we can create an indistinguishable and truthful model from the empirical distribution over the optimal responses from player $C$.

\subsubsection*{Learning a Function with Support Access}
In the classic setting when learning a function  $f:\{0,1\}\rightarrow\{0,1\}$, the learner receives random samples of form $(x,f(x))$ for a uniform $x\in\{0,1\}$. In this work, we also consider a function object in which we receive a random \emph{positive entry}. That is, the learner receives random $x$'s such that $f(x)=1$.  This type of random access is natural is certain situations, for example when we have information on the individuals that graduated some program, but not on those that did not. 

\begin{theorem}[Informal statement of \Cref{thm:rand-one}]\label{thm:rand-one-informal}
    Let $\alpha>0$, and let $f:\{0,1\}\rightarrow\set{0,1}$ be a function such that $\Pr[f(x)=1]=\alpha$. Let $\cD$ be a collection of distinguishers, each $D\in \cD$ associated with a set $S_D\subseteq[N]$ and accepts $x$ if $x\in S_D$. If there exists a weak agnostic learner for $\cD$, then there exists a learning algorithm $L$ running in time $\poly(n)$, that receives random elements from the set $\set{x|f(x)=1}$ and outputs a model $M$ that is indistinguishable from $f$ to all $D\in\cD$.
\end{theorem}

The theorem holds when there is a weak agnostic learner for the collection of distinguishers $\cD$ under the distribution of a random support element (i.e.\ random $x$ s.t.\ $f(x)=1$). We show in \Cref{sec:support-audit} that if a collection of distinguishers $\cD$ has a weak agnostic learner over the standard sample access distribution, and the learner is a statistical query algorithm, then there is a learner for $\cD$ also under the distribution of random support element. 

The proof of the theorem is similar to the classic boosting argument, with an additional step that the learner preforms of keeping the support size of the model approximately the same as support size of $f$. This step is necessary because under the distribution of a random support element, the boosting algorithm is only promised to work when the support sizes of $f$ and the model are approximately equal.

Learning an object under the distribution of random support element is potentially very useful in the case of \emph{sparse objects}. For a sparse function $f$, if we choose a uniform $x\in\{0,1\}$, then $f(x)=0$ with high probability, and a learner cannot hope to learn anything non-trivial with random samples of form $(x,f(x))$. Unfortunately, the above theorem does not hold for sparse functions, but in the next part we show how this theorem can be used to learn  different sparse objects -  sparse graphs.

\subsubsection*{Learning Sparse Graphs Without Dense Subgraphs}
Suppose $G=([N],E)$ is a graph represented by the $N^2$ length string of its adjacency matrix. In this representation, receiving a random edge from $G$ is equivalent to receiving a random support element from the function representing its adjacency matrix. Therefore, \Cref{thm:rand-one-informal} implies that we can learn a model for $G$ that is indistinguishable for a set of distinguishers $\cD$ that have a weak agnostic learner. The theorem only holds for functions $f$ with a constant fraction of $1$ entries, which corresponds to a dense graph. What about sparse graphs?

Learning a sparse graph, or a sparse object in general, is a very challenging task because of the huge domain. The weak regularity lemma \cite{Frieze} has error that is proportional to $N^2$, which is too much in the case of sparse graphs (an empty graph is indistinguishable from a sparse graph with this error). Therefore in the setting of a sparse graphs it is more natural to require an $\varepsilon$ error from the distinguisher under the distribution of receiving a random edge. Under this distribution, the error of the distinguishers scales with the number of edges.
We show a learner for a specific class of sparse graphs, those that have no dense subgraphs. We note that a random sparse graph has no dense subgraphs, so many graphs have this property. 

\begin{theorem}[Informal statement of \Cref{thm:sparse-graph}]\label{thm:sparse-graph-informal}
    Let $G=([N],E)$ be a sparse graph with no dense subgraphs.
    Let $\cD$ be a collection of distinguishers, each $D\in \cD$ associated with two sets $U_D,V_D\subseteq[N]$ and accepts an edge $(u,v)$ if $u\in U_D,v\in V_D$. If there exists a weak agnostic learner for $\cD$, then there exists a learning algorithm $L$ running in time $\polylog(N)$, that receives random edges from $G$ and outputs a model $M$ that is indistinguishable from $f$ to all $D\in\cD$.
\end{theorem}

The model that the learner outputs is \emph{dense}, i.e. the model outputs graphs with $\Theta(N^2)$ many edges. This is done because of technical reasons - to allow us to use rejection sampling when training the model. This brings us to the question, is there a dense graph that is indistinguishable from our sparse graph $G$? The answer to this question depends on $G$, and in \Cref{claim:no-dense-graph} we show that if a sparse graph $G$ has a very dense subgraph, then there is no dense graph that is indistinguishable from $G$. 

The proof of the theorem has two parts, in the first part we show that for every sparse graph $G$ with no dense subgraphs, there exists a dense graph $H$ that is indistinguishable from $G$. In this part of the proof we apply the strong regularity lemma for sparse graphs \cite{KR03,Scott2011} on $G$, and use the resulting partition to build the dense indistinguishable graph $H$. This part of the proof is existential, and we do know how to find $H$ efficiently, as the strong regularity lemma does not have an efficient algorithm for finding the partition. It is not possible to use the weak regularity lemma or its variants \cite{Frieze}, because its error is too large.  
In the second part of the proof, we reduce the learning $G$ to learning  $H$, and show that the resulted model $M$ is indistinguishable from $G$ to all distinguishers $D\in\cD$.

\subsubsection*{Other Results on Learning Generative Models}
In this work we also show indistinguishable models in several other settings
\begin{itemize}
    \item Let $f:\set{0,1}^n\rightarrow\set{0,1}^n$ be a function. Learning such function is harder than learning a binary function because the large domain makes $f$ a sparse object (when viewed as a graph for example it is an out-degree one graph). For such functions, we show that there exists a learner that given samples from the distribution $(x,f(x))$, outputs a model that is indistinguishable against the following set of distinguishers $\cD = \set{(S_D,j_D)| S_D\subset\set{0,1}^n,j_D\in [n]}$ such that $D=1\iff x\in S_D, f(x)_j=1$. This appears on \Cref{sec:bit-string-func}.
    \item In~\Cref{sec:dense} we apply the theorems for functions on the adjacency matrix of a dense graph $G=([N],E)$. For a set of distinguishers $\cD$ that have a weak agnostic leaner, we have an efficient learner that outputs an indistinguishable model when $G$ is:
    \begin{enumerate}
        \item A dense graph when the learner receives random adjacency matrix entries.
        \item A dense graph with a fixed total number of edges $m=\Theta(N^2)$.
        \item A directed graph with a fixed out-degree $m=\Theta(N)$.
        \item A dense graph when the learner receives random edges.
    \end{enumerate}
    \item For a directed graph $G=(\{0,1\}^n,E)$ with constant out-degree $d$, we can treat each of the $d$ outgoing degrees as a function $f_i:\set{0,1}^n\rightarrow\set{0,1}^n$. For the same set of distinguishers that we can handle in the case of a length-preserving functions (the first item in this list), we provide an efficient learner. See~\Cref{sec:const-out}.
    \item In the case of a uniform degree \emph{undirected graph}, we provide in \Cref{sec:uniform} a learning algorithm for an indistinghuishable model, albeit for a somewhat limited set of distinguishers.
\end{itemize}

\subsubsection*{Impossibility Results}
As we discussed earlier, our goal of learning a generative model is closely related to the goal in \cite{goldreich2010} of implementing huge random objects, but a key difference is that we assume the groundtruth is a single unknown object $B^*$, whereas \cite{goldreich2010} considers a known uniform distribution of objects.
This means that we need an additional learning procedure to collect information about $B^*$, and we show in \Cref{sec:impossibility} that our task of efficiently learning a generative model is only possible when the distinguisher class is efficiently learnable.

Besides the requirement of learning, our setting is more challenging than the setting in \cite{goldreich2010} in many other ways. We demonstrate this by another two impossibility results on fooling entry-access distinguishers and fooling stronger distinguishers than the model.
When we consider a pseudorandom function, $f_s$, the function is indistinguishable from the uniform distribution to distinguishers that have entry access to the function (allowed to ask for an arbitrary string $x$ and get $f_s(x)$). Furthermore, while  $f_s$ is computable in a fixed polynomial time, the distinguishers can run in any polynomial time (and under reasonable assumptions, even exponential time). \cite{goldreich2010} and subsequent work inherit these two properties - indistinguishability to distinguishers that are computationally more complex than the models and have entry access to the model. In Section~\ref{sec:impossibility} we argue that neither of these properties is achievable in our setting. 

For the impossibility of fooling distinguishers with entry access, in \Cref{thm:entry-hard} we give the example of a class $\cD$ that contains a distinguisher $D_x$ for every input $x\in \{0,1\}^n$ which queries the function value $f(x)$ for a function $f$ and outputs $\acc$ if and only if $f(x)=1$. We argue that every model $M$ that is indistinguishable from the true $f^*$ for the set of distinguishers $\cD$ has to be very close to $f^*$. Since the size of $f^*$ is exponential and $f^*$ is unknown, no efficient learner can output a model that is close to $f^*$.
We also show an example, using an idea from \cite{TrevisanTV09}, of a distinguisher and a true function $f^*$, such that the distinguisher can tell apart $f^*$ from any model $M$ with a low complexity compared to the distinguisher (\Cref{thm:strong-model}). This highlights the fact that in our setting, the generative model and the learner constructing the model have to be computationally comparable or stronger than the distinguishers.
\subsection{Related Work}

As mentioned in \Cref{sec:overall-goal}, \cite{goldreich2010} introduced the problem of creating an indistinguishable implementation of a random object. 
\cite{goldreich2010} as well as follow-up works \cite{NaorNT05,NaorN07} also present a collection of positive results for dense and sparse graphs or functions with a variety of truthfulness conditions and access models of the distinguishers. 

The connection between generative models and indistinguishability has been manifested through the invention of generative adversarial networks (GANs) \cite{NIPS2014_5ca3e9b1,arora2017generalization}. Intuitively, a GAN is trained to imitate a distribution of objects (say images). The generator is trained in concert with a discriminator that could be interpreted as a distinguisher. Through a sequence of rounds, the generator is trained to fool the discriminator which is then trained to fail the generator. GANs highlight the connection between  generative models and indistinguishability \cite{IMP17}, but they do not naturally fall into our framework as they are more directly described in terms of indistinguishability of two distributions. 

The connection between indistinguishability and learning theory has been established in many previous works (e.g.\ \cite{TrevisanTV09} applies the boosting technique from learning theory).
More recently, in the context of algorithmic fairness, the relation between learning theory and indistinguishability has been dramatically expanded in the notions of multicalibration and outcome indistinguishability~\cite{hebert2018multicalibration,dwork2021outcome}, in applications to learning and statistical inference through the notions of omnipredictors and universal adapatability~\cite{GopalanKRSW22,doi:10.1073/pnas.2108097119,hu2022omnipredictors,gopalan_et_al:LIPIcs.ITCS.2023.60,kim_et_al:LIPIcs.ITCS.2023.79} and in the emergence of research uncovering intricate and exciting connections while studying the sample complexity of indistinguishability from a learning-theoretic perspective~\cite{pmlr-v167-hu22a,hu_et_al:LIPIcs.ITCS.2023.72}.

It is possible to view our learning setting as a $2$-players zero-sum game, between the learner and the distinguishers, in which the learner's goal is to output a model for an indistinguishable object and the distinguishers try to tell apart the input and the model. In this setting, there is a relation between min-max theorems and regularity-lemma theorems. Such theorems prove that it is possible to express a complex object $f$ by a function of a few simpler objects $g_1,\ldots g_t$ that, in our setting, represent the distinguishers \cite{TrevisanTV09,vadhan2013uniform}. There have been works improving the parameters and also using such theorems for applications in cryptography \cite{vadhan2013uniform,MR3183555,MR3591815,skorski2016subgradient,MR3655523,MR3794838}. In this work, our setting is slightly different, as we assume that the object $f$ is complex and unknown, and the learner has to learn it. The proof of \Cref{thm:sparse-graph-informal} has an intermediate step that is existential and has a similar structure to a weak regularity lemma theorems, but since the required error there is too small, we derive it from the sparse strong regularity lemma.

\section{Preliminaries}

Throughout the paper, we are interested in learning objects such as functions and graphs, and we are particularly interested when these objects have exponential sizes (e.g.\ functions with exponentially large domains and graphs with exponentially many vertices and edges). We typically use $B$ to denote an object, and use $\cB$ to denote a class of objects. 
We view an object $B$ as a function $B:Q\to \Delta_A$ that maps a query $q\in Q$ to a distribution $B(q)$ over answers in $A$. 

\paragraph{Functions.}
When the object is a function $f:X\rightarrow Y$, we consider three access types. For sample access, $B$ returns a random pair $(x,f(x))$. For support access, it returns a random $x$ such that $f(x)=1$. For entry access, upon querying $x$, $B$ returns $f(x)$. 
\begin{definition}[Function-induced sample-access object]
\label{def:sample-function}
Let $f:X\to Y$ be a function and let $B:Q \to \Delta_A$ be an object. We say $B$ is the \emph{sample-access} object \emph{induced} by $f$ if $Q = \{\bot\}, A = X\times Y$, and $B(\bot)$ is the distribution of $(x,f(x))\in A$ where $x$ is drawn uniformly from $X$. 
\end{definition}
\begin{definition}[Function-induced support-access object]
\label{def:support-function}
Let $f:X\to \{0,1\}$ be a binary function. We define the support of $f$ to be $\supp(f):= \{x\in X: f(x) = 1\}$. 
Let $B:Q\to \Delta_A$ be an object.
Assuming $\supp(f)\ne \emptyset$,
we say $B$ is
the \emph{support-access} object induced by $f$ if $Q = \{\bot\}, A= X$, and $B(\bot)$ is the uniform distribution over $\supp(f)\subseteq X$.
\end{definition}
\begin{definition}[Function-induced entry-access object]
\label{def:entry-function}
Let $f:X\to Y$ be a function and let $B:Q \to \Delta_A$ be an object. We say $B$ is the \emph{entry-access} object \emph{induced} by $f$ if $Q = X, A = Y$, and for every $q\in Q$, $B(q)$ is the singleton distribution such that $a\sim B(q)$ equals to $f(q)$ deterministically. 
\end{definition}
In this paper, we show positive results for learning generative models of functions with sample access and support access (\Cref{sec:func-learning}) whereas we show impossibility results for entry access (\Cref{sec:impossibility}). This separation is mainly because entry access makes the distinguishers stronger and thus makes indistinguishability harder to achieve (see \Cref{def:distinguisher,def:indistinguishability} below).
\paragraph{Graphs.}
For a graph $G=(V,E)$ where we assume $V$ has exponential size, we define two access types, sample-access which corresponds to a random adjacency matrix entry, and support-access which corresponds to a random edge in the graph.
\begin{definition}[Graph-induced sample-access object]
\label{def:sample-graph}
Let $G = (V,E)$ be a directed or undirected graph and let $B:Q\to \Delta_A$ be an object. We say $B$ is the \emph{sample-access} object induced by $G$ if $Q = \{\bot\},A = V\times V\times\{0,1\}$, and $B(\bot)$ is the distribution of $(u,v,y)\in A$ where $(u,v)$ is drawn uniformly from $V\times V$, $y = 1$ if $(u,v)\in E$ and $y = 0$ otherwise.
\end{definition}

\begin{definition}[Graph-induced support-access object]
\label{def:support-graph}
Let $G = (V,E)$ be a directed or undirected graph and let $B:Q\to \Delta_A$ be an object. Assuming $E\ne \emptyset$, we say $B$ is the \emph{support-access} object induced by $G$ if $Q = \{\bot\}, A = V\times V$, and $B(\bot)$ is the uniform distribution over $E\subseteq V\times V$.
\end{definition}

\paragraph{Indistinguishability.} Each learner we design in this paper has access to a ground-truth object $B^*$, and it aims to output an object $B$ that is \emph{indistinguishable} from $B^*$. In many cases, the learner does not just output a single object $B$, but a distribution over objects, and we refer to such distributions as \emph{models}. Below we formally define the notion of indistinguishability.

\begin{definition}[Distinguisher]
\label{def:distinguisher}
A \emph{distinguisher} $D$ is an algorithm that when given access to an object $B:Q\to \Delta_A$, outputs $\acc$ or $\rej$. That is, the distinguisher is allowed to make queries $q\in Q$ to the model, and for each query $q$ the distinguisher receives an answer $a\in A$ drawn independently from $B(q)\in \Delta_A$. We allow the distinguisher $D$ itself to be randomized, and we use random variable $D^B$ to denote the output of the distinguisher $D$ in $\{\acc, \rej\}$ when given access to $B$.
\end{definition}

\begin{definition}[Indistinguishability]
\label{def:indistinguishability}
Let $B^*:Q\to \Delta_A$ be an object, and let model $M$ be a distribution over objects $B:Q\to \Delta_A$. We say model $M$ is \emph{$\varepsilon$-indistinguishable} from object $B^*$ w.r.t.\ a distinguisher $D$ if
\begin{equation}
\label{eq:ind}
|\Pr[D^{B^*}= \acc] - \E_{B\sim M}[\Pr[D^{B}= \acc]]|\le \varepsilon.
\end{equation}
We say model $M$ is $\varepsilon$-indistinguishable from object $B^*$ w.r.t.\ a class $\cD$ of distinguishers if \eqref{eq:ind} holds for every $D\in \cD$.
\end{definition}

\paragraph{Truthfulness.} In addition to indistinguishability, another desirable property of a model is \emph{truthfulness} introduced in \cite{goldreich2010}. Truthfulness requires every object generated from the model to satisfy a certain (usually global) property which we formalize using an object class $\cB$:
\begin{definition}[Truthfulness]
We say a model $M$ is \emph{truthful} w.r.t.\ an object class $\cB$ if
\[\Pr\nolimits_{B\sim M}[B\in \cB] = 1.
\]
\end{definition}

\paragraph{Implementations.}
Our goal is to design efficient learners, and thus we cannot expect the learner to output a model $M$ explicitly, especially when the objects drawn from $M$ are huge. Instead, our learner outputs an efficient \emph{implementation} of a model, defined as follows.
\begin{definition}[Ordinary Implementation]
\label{def:ordinary-imp}
For $\ell\in \bZ_{\ge 0}$,
let $T$ be a randomized algorithm that takes $(r,q)\in \{0,1\}^\ell \times Q$ as input, and outputs $T(r,q)\in A$.
We say $T$ is an \emph{ordinary implementation} of a model $M$ with \emph{seed length} $\ell$ if for every \emph{seed} $r\in \{0,1\}^\ell$ there exists an object $B_r:Q\to\Delta_A$ such that
\begin{enumerate}
    \item for every $q\in Q$, $T(r,q)$ is distributed according to $B_r(q)$, where the randomness in $T(r,q)$ comes from the internal randomness in algorithm $T$;
    \item $B_r$ is distributed according to $M$ when $r$ is drawn uniformly from $\{0,1\}^\ell$.
\end{enumerate}
\end{definition}

While our goal is to output an ordinary implementation with a polynomial-length seed, following the approach in \cite{goldreich2010}, it is more convenient to first build implementations using a random oracle and then transform the implementation to an ordinary one using \Cref{lm:random-oracle}. 

\begin{definition}[Random-Oracle Implementation]
\label{def:random-imp}
Let $T$ be a randomized algorithm that takes a function $r:\{0,1\}^*\to\{0,1\}$ as an oracle. On an input $q\in Q$, the algorithm $T$ outputs $T^r(q)\in A$. We say $T$ is a \emph{random-oracle implementation} of a model $M$ if for every $r:\{0,1\}^*\to\{0,1\}$ there exists an object $B_r:Q\to \Delta_A$ such that
\begin{enumerate}
    \item for every $q\in Q$, $T^r(q)$ is distributed according to $B_r(q)$, where the randomness in $T^r(q)$ comes from the internal randomness in algorithm $T$;
    \item $B_r$ is distributed according to $M$ when $r$ is a uniformly random function from $\{0,1\}^*$ to $\{0,1\}$.
\end{enumerate}
\end{definition}

\begin{lemma}[Theorem 2.9 in \cite{goldreich2010}]
\label{lm:random-oracle}
Suppose that one-way functions exist.
There exists an algorithm $H$ with the following properties. 
Let $\cD$ be a class of distinguishers where each $D\in \cD$ is a circuit of size at most $W$ for some $W \ge 1$. Let $T$ be a random-oracle implementation of a model $M$ with circuit complexity at most $W$. Given $W,T$ and an arbitrary $\varepsilon\in (0,1)$ as input, the algorithm $H$ runs in time $\poly(W,1/\varepsilon)$ and outputs an ordinary implementation $T'$ of a model $M'$ where $T'$ has seed length and circuit complexity both being $\poly(W,1/\varepsilon)$, and
\[
|\E_{B'\sim M'}\Pr[D^{B'} = \acc] - \E_{B\sim M}\Pr[D^{B} = \acc]|\le \varepsilon \quad \text{for every }D\in \cD.
\]
Moreover, if $M$ is truthful w.r.t.\ an object class $\cB$, then $M'$ is also truthful w.r.t. $\cB$.
\end{lemma}
\Cref{lm:random-oracle} can be proved by using a pseudorandom function to emulate the random oracle.

\paragraph{Learning.}
We describe the learners we aim to design in the definition below.
\begin{definition}[Learner]
Let $\cB$ be a class of objects $B:Q\to \Delta_A$ and $\cD$ be a class of distinguishers. An \emph{$(\varepsilon,\delta)$-learner} L for the class $\cB$ w.r.t.\ $\cD$ is an algorithm with the following properties. For any $B^*\in \cB$, given access to $B^*$, the learner outputs an implementation $T$ of a model $M$ such that with probability at least $1-\delta$, $M$ is $\varepsilon$-indistinguishable from $B^*$ w.r.t.\ $\cD$. 
\end{definition}

\paragraph{Other Notations.}
For $v\in \bR$, we define $\numcap(v)$ by capping its value into $[0,1]$, i.e.,
\[
    \numcap(v) = \begin{cases}
        v, \quad & \text{if }0\leq v\leq 1;\\
        1, \quad &\text{if }v>1;\\
        0, \quad &\text{if }v<0.
    \end{cases}
\]
Given a list of values $(v_1,\ldots,v_t)$ we define $\lcap(v_1,\ldots,v_t)$ by summing over the list and capping the value to $[0,1]$ in every iteration. We formally define it recursively: 
\begin{align}
\lcap(v_1) & = \numcap(v_1),\nonumber \\
\lcap(v_1,\ldots, v_t) & = \numcap(\lcap(v_1,\ldots,v_{t-1})+v_t). \label{eq:lcap}
\end{align}

\section{Learning Functions with Exponentially Large Domains}\label{sec:func-learning}
The goal of this section is to efficiently learn a generative model that is indistinguishable from a target function $f^*:X\to Y$ to a class $\cD$ of distinguishers. 
We allow the domain $X$ of the function to have exponential size $N:= |X|$, and require our learner to run in time $\polylog(N)$. 
This means that the learner cannot read the entire function $f^*$, and can only access it via random sample. 
Throughout the paper, our learners output an efficient random-oracle implementation $T$ of a model $M$, which can be turned in to an efficient ordinary implementation by \Cref{lm:random-oracle}.

\subsection{Learning Sample-Access Binary Functions}
\label{sec:sample-function}
We start by studying the case where the target object $B^*$ is the sample-access object induced by a binary function $f^*:X\to\{0,1\}$ (\Cref{def:sample-function}). 
We assume that every distinguisher $D\in \cD$ satisfies the following: when given access to a sample-access object $B$ induced by a function $f:X\to\{0,1\}$, the distinguisher asks a single query $\bot$, receives an answer $a = (x,y)\sim B(\bot)$, and outputs $D(x,y)\in \{\acc, \rej\}$. We allow the distinguisher itself to be randomized, and each distinguisher $D$ defines a function $g_D:X\to [-1,1]$ such that
\[
g_D(x) = \Pr[D(x,1) = \acc] - \Pr[D(x,0) = \acc] \quad \text{for every }x\in X.
\]
We use the following claim to relate a distinguisher $D\in \cD$ to the function $g_D:X\to[-1,1]$:
\begin{claim}
\label{claim:distinguisher-function}
For every distinguisher $D\in \cD$,
the following equation holds for every $x\in X$ and $y_1,y_2\in \{0,1\}$:
\[
\Pr[D(x,y_1) = \acc] - \Pr[D(x,y_2) = \acc] = y_1g_D(x) - y_2g_D(x).
\]
\end{claim}
The claim can be easily proved by considering the four possible choices of $(y_1,y_2)\in \{0,1\}\times\{0,1\}$.

As we show later in \Cref{sec:impossibility}, it is necessary to impose certain learnability assumptions on the distinguishers. To that end, we assume that there is an \emph{auditor} for the function class $\cG:=\{g_D:D\in \cD\}$, defined as follows:
\begin{definition}[Auditor]
\label{def:auditor-sample-function}
Let $\cD$ and $\cG$ be defined as above. 
We say an algorithm $\Lambda$ is an $(\varepsilon,\gamma,\delta)$-auditor for $\cG$ if it satisfies the following property. Given access to a sample-access object $B^*$ induced by a function $f^*:X\to\{0,1\}$ and taking a predictor $p:X\to[0,1]$ as an oracle, if there exists $g\in \cG$ such that
\[
|\E[f^*(x)g(x)] - \E[p(x)g(x)]| \ge \varepsilon,
\]
then $\Lambda$ outputs $\hat g:X\to[-1,1]$ satisfying the following with probability at least $1-\delta$:
\[
\E[f^*(x)\hat g(x)] - \E[p(x)\hat g(x)] \ge \gamma.
\]
\end{definition}

The auditor defined above can be viewed as a weak agnostic learner for the class $\cG$. When the domain $X$ is $\{0,1\}^n$ with size $N = 2^n$, many classes allow efficient auditors that run in time $\poly(n) = \polylog(N)$. Using an auditor for $\cG$, we prove the following theorem:
\begin{theorem}
\label{thm:sample-function}
Let the distinguisher class $\cD$ and the function class $\cG$ be defined above. Let $\varepsilon,\gamma,\delta,\delta'\in (0,1/2)$ be parameters satisfying $\gamma \le \varepsilon$ and $\delta' \le c\delta\gamma^2$ for a sufficiently small absolute constant $c > 0$. Let $\cB$ be the class of sample-access objects induced by binary functions $f:X\to\{0,1\}$. Let $\Lambda$ be an $(\varepsilon,\gamma,\delta')$-auditor for $\cG$ (\Cref{def:auditor-sample-function}). Then there exists an $(\varepsilon,\delta)$-learner $L$ for $\cB$ w.r.t.\ $\cD$.
Moreover, if the auditor $\Lambda$ runs in time at most $W_1$ and always outputs a function with circuit size at most $W_2$, then the learner $L$ runs in time $\poly(\gamma^{-1},\log(\delta^{-1}),W_1)$ and always outputs implementations with circuit complexity $\tilde O(\gamma^{-2}W_2)$.
\end{theorem}
We apply the following result in the algorithmic fairness literature to prove \Cref{thm:sample-function}:
\begin{lemma}[Learning a multiaccurate predictor \cite{TrevisanTV09,hebert2018multicalibration,kim2019multiaccuracy}]
\label{lm:predictor}
Let the distinguisher class $\cD$ and the functions class $\cG$ be defined as above. 
Let $\varepsilon,\gamma,\delta,\delta'\in (0,1/2)$ be parameters satisfying $\gamma \le \varepsilon$ and $\delta' \le c\delta\gamma^2$ for a sufficiently small absolute constant $c > 0$.
Let $\Lambda$ be an $(\varepsilon,\gamma,\delta')$-auditor for $\cG$.
Given access to a sample-access object $B^*$ induced by a function $f^*:X\to\{0,1\}$, there is a learner $L$ that outputs a predictor $p:X\to[0,1]$ such that with probability at least $1-\delta$,
\begin{equation}
\label{eq:predictor}
|\E[f^*(x)g(x)] - \E[p(x)g(x)]| \le \varepsilon \quad \text{for every }g\in \cG.
\end{equation}
Moreover, if the auditor $\Lambda$ runs in time at most $W_1$ and always outputs a function with circuit size at most $W_2$, then the learner $L$ runs in time $\poly(\gamma^{-1},\log(\delta^{-1}),W_1)$ and always outputs implementations with circuit complexity $\tilde O(\gamma^{-2}W_2)$.
\end{lemma}
\begin{proof}[Proof of \Cref{thm:sample-function}]
Let $B^*$ be the sample-access object induced by a binary function $f^*:X\to\{0,1\}$.
It suffices to turn the predictor $p$ from \Cref{lm:predictor} into an efficient random-oracle implementation $T$ of a model $M$ that is $\varepsilon$-indistinguishable from $B^*$ w.r.t.\ $\cD$. Let us consider a random function $f:X\to\{0,1\}$ such that for every $x\in X$, the function value $f(x)$ is drawn independently from $\ber(p(x))$. By \Cref{claim:distinguisher-function}, for every distinguisher $D\in \cD$, we have
\begin{align}
& |\E_x[\Pr[D(x,f^*(x)) = \acc]] - \E_{f,x}[\Pr[D(x,f(x)) = \acc]]|\notag\\
= {} & |\E_x[f^*(x)g_D(x)] - \E[p(x)g_D(x)]|.\label{eq:sample-function-1}
\end{align}
Therefore, if we choose the model $M$ to be the distribution over the sample-access object $B_f$ induced by the random function $f$, then \eqref{eq:predictor} and \eqref{eq:sample-function-1} imply that $M$ is $\varepsilon$-indistinguishable from $B^*$ w.r.t.\ $\cD$. To implement $M$ efficiently using a random oracle $r:\{0,1\}^*\to\{0,1\}$, for each $x\in X$, we designate a set $S_x\subseteq\{0,1\}^*$ so that $S_x\cap S_{x'} = \emptyset$ for $x\ne x'$.
We first generate a random $x\in X$ uniformly from $X$, and then draw $y\sim \ber(p(x))$ using the randomness from $S_x$. We return the answer $(x,y)$.
This is an implementation for $M$ for the following reasons. For every fixed $r$, there exists a function $f:X\to\{0,1\}$ such that the answer $(x,y)$ always satisfies $y = f(x)$ and the distribution of $x$ is uniform in $X$. Moreover, when $r$ is random, the conditional distribution of $y$ given $x$ is $\ber(p(x))$.
\end{proof}

\subsection{Truthful Learning That Preserves Support Size}
Some desirable properties of a generative model are global and cannot be enforced using computationally bounded distinguishers alone. 
This motivated \cite{goldreich2010} to introduce the notion of \emph{truthfulness} that ensures such global properties beyond indistinguishability.
Here we focus on a natural global property of a binary function $f:X\to \{0,1\}$: the size of its support $\supp(f):= \{x\in X:f(x) = 1\}$. 
The model $M$ we create in \Cref{sec:sample-function} using the multiaccurate predictor $p$ from \Cref{lm:predictor} may generate functions $f$ with support size different from the target function $f^*$. Indeed, even if $\sum_{x\in X}p(x) = |\supp(f^*)|$, a random function $f$ with each entry $f(x)$ independently drawn from $\ber(p(x))$ is not guaranteed to satisfy $|\supp(f)| = |\supp(f^*)|$. Now we show an efficient learner that outputs truthful models that preserve the support size of the generated functions.

\label{sec:fixed-weight}
\begin{theorem}
\label{thm:fixed-weight}
Let the distinguisher class $\cD$ and the function class $\cG$ be defined as in \Cref{sec:sample-function}. Let $\varepsilon,\gamma,\delta,\delta'\in (0,1/2)$ be parameters satisfying $\gamma \le \varepsilon$ and $\delta' \le c\delta\gamma^2$ for a sufficiently small absolute constant $c > 0$. Let $\cB$ be the class of sample-access objects induced by binary functions $f:X\to\{0,1\}$ satisfying $|\supp(f)| = k$, where $X = \{0,1\}^n$. Let $\Lambda$ be an $(\varepsilon,\gamma,\delta')$-auditor for $\cG$ (\Cref{def:auditor-sample-function}). Then there exists an $(\varepsilon,\delta)$-learner $L$ for $\cB$ w.r.t.\ $\cD$ and the learner always outputs a random-oracle implementation of a model $M$ that is truthful w.r.t.\ $\cB$.
Moreover, if the auditor $\Lambda$ runs in time at most $W_1$ and always outputs a function with circuit size at most $W_2$, then the learner $L$ runs in time $\poly(n,\gamma^{-1},\log(\delta^{-1}),W_1)$ and always outputs an implementation with circuit complexity $\poly(n,\gamma^{-1},W_2)$.
\end{theorem}
We prove \Cref{thm:fixed-weight} using \Cref{lm:predictor} and the following lemma:
\begin{lemma}[Fixed-weight implementation from a predictor]
\label{lm:fixed-weight}
Define $X = \{0,1\}^n$ for a positive integer $n$.
Let $k$ be an integer satisfying $0 \le k \le |X|$ and let $\varepsilon\in (0,1)$ be a parameter. There exists a $\poly(n,\varepsilon^{-1})$-time implementation $T$ such that $T$ takes an arbitrary predictor $p:X\to [0,1]$ as an oracle and $T$ implements a model $M$ such that each object $B\sim M$ is the sample-access object induced by a binary function $f_B:X\to\{0,1\}$. Moreover, for every predictor $p:X\to [0,1]$, the resulting model $M$ satisfies
\begin{enumerate}
    \item $\Pr_{B\sim M}[\sum_{x\in X}f_B(x) = k] = 1$;
    \item $\sum_{x\in X}|\E_{B\sim M}f_B(x) - p(x)| \le |k - \sum_{x\in X}p(x)| + \varepsilon|X|$.
\end{enumerate}
\end{lemma}
We first prove \Cref{thm:fixed-weight} using \Cref{lm:fixed-weight} and then prove \Cref{lm:fixed-weight}.
\begin{proof}[Proof of \Cref{thm:fixed-weight}]
We start with the observation that using $\Lambda$, we can build an auditor $\Lambda'$ for $\cG\cup \{g_1\}$, where $g_1$ is the constant $1$ function: $g_1(x) = 1$ for every $x\in X$. This is because we can accurately estimate $\E[f^*(x)g_1(x)] - \E[p(x)g_1(x)]$ using random examples $(x,f^*(x),p(x))$ for $x$ drawn uniformly from $X$.

Let $B^*$ be the sample-access object induced by a function $f^*:X\to\{0,1\}$ satisfying $\supp(f^*) = k$. By \Cref{lm:predictor}, we can learn a predictor $p:X\to\{0,1\}$ such that
\begin{equation}
\label{eq:fixed-weight-thm-0}
|\E_x[f^*(x)g(x)] - \E_x[p(x)g(x)]| \le \varepsilon/4 \quad \text{for every }g\in \cG \cup \{g_1\}.
\end{equation}
In particular, choosing $g = g_1$, we have
\[
\left|k - \sum_{x\in X}p(x)\right| = |X|\cdot |\E[f^*(x)g_1(x)] - \E[p(x)g_1(x)]| \le \varepsilon |X|/4.
\]

We apply \Cref{lm:fixed-weight} on $p$ to obtain an implementation $T$ of a model $M$ such that each object $B\sim M$ is the sample-access object induced by a binary function $f_B:X\to\{0,1\}$. \Cref{lm:fixed-weight} allows us to ensure that
\begin{align}
\Pr\nolimits_{B\sim M}[|\supp(f_B)| = k] & = 1,\label{eq:fixed-weight-thm-1}\\
\sum_{x\in X} |\E_{B\sim M}f_B(x) - p(x) | & \le \varepsilon |X|/2. \label{eq:fixed-weight-thm-2}
\end{align}

 It remains to prove that $M$ is $\varepsilon$-indistinguishable from $B^*$ w.r.t.\ $\cD$ and that $M$ is truthful w.r.t.\ $\cB$. The truthfulness follows immediately from \eqref{eq:fixed-weight-thm-1}. We prove indistinguishability below. For every $g\in \cG$, we have
\begin{align}
|\E_x[p(x)g(x)] - \E_{B\sim M}\E_x[f_B(x)g(x)]| & \le \Big|\E_x\Big[g(x)\E_{B\sim M}[f_B(x) - p(x)]\Big]\Big|\notag \\
& \le \E_x|\E_{B\sim M}f_B(x) - p(x)|\notag \\
& \le \varepsilon/2, \label{eq:fixed-weight-thm-3}
\end{align}
where the last inequality follows from \eqref{eq:fixed-weight-thm-2}.
Combining \eqref{eq:fixed-weight-thm-0} and \eqref{eq:fixed-weight-thm-3}, for every $g\in \cG$,
\[
|\E_x[f^*(x)g(x)] - \E_{B\sim M}\E_x[f_B(x)g(x)]| \le \varepsilon.
\]
By \Cref{claim:distinguisher-function}, this implies that $M$ is $\varepsilon$-indistinguishable from $B^*$ w.r.t.\ $\cD$.
\end{proof}
In the following sub-sections, we prove \Cref{lm:fixed-weight} starting with an inefficient implementation and building to more and more efficient ones, where we use $N = 2^n$ to denote the domain size $|X|$. For convenience, we sometimes take $[N]:= \{1,\ldots,N\}$ to be our domain in place of $X = \{0,1\}^n$, in which case we allow $N$ to be not necessarily a power of $2$.

\subsubsection{Step 1: Exact Model with \texorpdfstring{$2^{O(N)}$}{2\^{}O(N)} Time Implementation}
We prove the following lemma which can be viewed as a version of \Cref{lm:fixed-weight} with $\alpha = \varepsilon = 0$ and without efficiency requirement for the model:
\begin{lemma}
\label{lm:fixed-weight-exact}
Let $p:[N] \to [0,1]$ be a predictor satisfying $\sum_{x\in [N]}p(x) = k$ for some integer $k$. There exists a distribution $\mu$ over \emph{binary} functions $f:[N] \to \{0,1\}$ such that
\begin{enumerate}
\item $\Pr_{f\sim\mu}[\sum_{x\in [N]}f(x) = k] = 1$, and
\item for every $x\in [N]$, $\E_{f\sim\mu}[f(x)] = p(x)$.
\end{enumerate}
\end{lemma}
For two functions $p,p':[N] \to \bR$, we define their inner product to be $\langle p,p' \rangle := \sum_{x\in [N]}p(x)p'(x)$. 
Our proof of \Cref{lm:fixed-weight-exact} is based on the following claim:
\begin{claim}
\label{claim:fixed-weight-helper}
Let $N,k$ be integers satisfying $N > 0$ and $0 \le k \le N$. Let $p:[N]\to[0,1]$ be a function satisfying $\sum_{x\in [N]}p(x) = k$.
For any function $w:[N] \to \bR$ with function values sorted as $w(\sigma(1)) \ge \cdots \ge w(\sigma(N))$ for a bijection $\sigma:[N]\to [N]$, let $f:[N]\to\{0,1\}$ be the function satisfying $f(\sigma(1)) = \cdots = f(\sigma (k)) = 1$ and $f(\sigma(k+1)) = \cdots = f(\sigma(N)) = 0$. Then
\[
\langle w,f\rangle \ge \langle w, p\rangle.
\]
\end{claim}
\begin{proof}
By rearranging the coordinates, it is w.l.o.g.\ to assume that $\sigma(i) = i$ for every $i\in [N]$. Defining $w(N+1) =0$, we have
\begin{align*}
\langle w,p\rangle&  = \sum_{i\in [N]} w(i)p(i)\\
& = \sum_{i\in [N]}\sum_{j = i}^N (w(j) - w(j+1))p(i)\\
& = \sum_{j\in [N]}\sum_{i\in [j]}(w(j) - w(j+1)p(i)\\
& = \sum_{j\in [N]}(w(j) - w(j+1))\sum_{i\in [j]}p(i)
\end{align*}
and similarly
\[
\langle w,f\rangle = \sum_{j\in [N]}(w(j) - w(j+1))\sum_{i\in [j]}f(i).
\]
It is clear that $\sum_{i\in [j]}p(i) \le \min\{j,k\} = \sum_{i\in [j]}f(i)$ for every $j = 1,\ldots,N$, and additionally $\sum_{i\in [N]}p(i) = k = \sum_{i\in [N]}f(i)$. Therefore,
\[
\langle w,p\rangle =\sum_{j\in [N]}(w(j) - w(j+1))\sum_{i\in [j]}p(i) \le \sum_{j\in [N]}(w(j) - w(j+1))\sum_{i\in [j]}f(i) = \langle w,f\rangle.
\qedhere
\]
\end{proof}
\begin{proof}[Proof of \Cref{lm:fixed-weight-exact}]
Let $F$ denote the set of functions $f:[N]\to\{0,1\}$ satisfying $|\supp(f)| = k$.
We first show that $p$ belongs to the convex hull of $F$. Assuming that this is not the case, by the hyperplane separation theorem, there exists $w:[N] \to \bR$ and $b\in \bR$ such that $\langle w,p\rangle - b > 0$ and $\langle w, f\rangle - b < 0$ for every $f\in F$. By \Cref{claim:fixed-weight-helper}, there exists $f\in F$ such that $\langle w,f\rangle \ge \langle w,p\rangle$, and thus $\langle w,f\rangle - b \ge \langle w,p\rangle - b > 0$, a contradiction.

We have proved that $p$ is in the convex hull of $F$. This means that there exists a probability distribution $\mu$ over $F$ such that $\E_{f\sim\mu}[f(x)] = p(x)$ for every $x\in [N]$, as desired.
\end{proof}

Clearly, the distribution $\mu$ can be explicitly computed using a linear program in time $2^{O(N)}$ where the variables are the probability mass on each $f:[N]\to \{0,1\}$.

\subsubsection{Step 2: Indistinguishable Model with \texorpdfstring{$\poly (N)$}{poly(N)} Time Implementation}
We prove the following variant of \Cref{lm:fixed-weight-exact} where we allow a small $\varepsilon$ error but require a $\poly(N)$ time algorithm to sample from the distribution $\mu$.
\begin{lemma}
\label{lm:fixed-weight-mw}
Let $N > 1$ be a positive integer, and let $p:[N] \to [0,1]$ be a predictor satisfying $\sum_{x\in [N]}p(x) = k$ for some integer $k$. For any $\varepsilon\in (0,1/2)$, there exists a distribution $\mu$ over \emph{binary} functions $f:[N] \to \{0,1\}$ such that
\begin{enumerate}
\item $\Pr_{f\sim\mu}[\sum_{x\in [N]}f(x) = k] = 1$, and
\item for every $x\in [N]$, $|\E_{f\sim\mu}[f(x)] - p(x)| \le \varepsilon$.
\end{enumerate}
Moreover, $\mu$ can be taken as the uniform distribution over $f\sps 1,\ldots,f\sps T$ for $T = O(\varepsilon^{-2}\log N)$, and the list $f\sps 1,\ldots,f\sps T$ can be computed in time $\tilde O(\varepsilon^{-2}N)$ given $p$ as input.
\end{lemma}
We prove \Cref{lm:fixed-weight-mw} using the multiplicative weights algorithm. 
We need the following notations. We define $\Delta_N$ to be the set of all functions $w:[N]\to\bR_{\ge 0}$ satisfying $\sum_{x\in [N]}w(x) = 1$. For two functions $w,w'\in \Delta_N$, assuming $w'(x) > 0$ for every $x\in [N]$, we define $\dkl(w\|w')= \sum_{x\in [N]}w(x)\ln(w(x)/w'(x))$ with the convention that $0\ln 0 = 0$. For a function $r:[N]\to \bR$, we define $\|r\|_\infty:= \max_{x\in [N]}|r(x)|$. For functions $r,r':[N]\to\bR$, we define $r\diamond r':[2N]\to\bR$ to be the function such that
\[
(r\diamond r')(x) = \begin{cases}
r(x), & \text{if }x\le N;\\
r'(x - N),  & \text{if }x > N.
\end{cases}
\]
The definition of $r\diamond r'$ allows us to state the following basic fact that will be useful:
\begin{claim}
\label{claim:infinity-norm}
For any function $r:[N]\to \bR$, there exists $w\in \Delta_{2N}$ such that $\langle w, r \diamond(-r)\rangle = -\|r\|_\infty$.
\end{claim}
The claim above allows us to express $-\|r\|_\infty$ as an inner product. Note that it is crucial to extend $r$ to $r\diamond (-r)$. The claim would not hold if we instead consider inner products of the form $\langle w,r\rangle$ for $w\in \Delta_N$ because of the restriction that $w(x) \ge 0$ for every $w\in \Delta_N$ and $x\in [N]$.

The following is a standard lemma for analyzing the multiplicative weights algorithm:
\begin{lemma}[Multiplicative Weights]
\label{lm:mw}
Let $w\sps 1\in \Delta_N$ and $r:[N]\to\bR$ be two functions. Assume $w\sps 1(x) > 0$ for every $x\in [N]$. Define $\hat w:[N]\to\bR$ and $w\sps 2\in \Delta_N$ such that for every $x\in [N]$,
\begin{align*}
\hat w(x) & = w\sps 1(x)e^{-r(x)},\\
w\sps 2(x) & = \hat w(x)/\sum_{x'\in [N]}\hat w(x').
\end{align*}
Then for every $w\in \Delta_N$
\[
\langle w\sps 1 - w, r\rangle \le \dkl(w\|w\sps 1) - \dkl(w\|w\sps 2) + \frac 12\|r\|_\infty^2.
\]
\end{lemma}
\begin{proof}[Proof of \Cref{lm:fixed-weight-mw}]
We initialize $w\sps 1\in \Delta_{2N}$ such that $w\sps 1(x) = 1/(2N)$ for every $x\in [2N]$. In the $i$-th iteration, we perform the following computation. We break $w\sps i$ into two functions $w\sps{i,+},w\sps{i,-}:[N]\to [0,1]$ such that $w\sps i = w\sps{i,+} \diamond w\sps{i,-}$. Using \Cref{claim:fixed-weight-helper}, we can efficiently compute $f\sps i\in \{0,1\}^N$ such that $|\supp(f\sps i)| = k$ and 
\begin{equation}
\label{eq:fixed-weight-mw-1}
\langle w\sps{i,+} - w\sps{i,-}, f\sps i\rangle \ge \langle w\sps{i,+} - w\sps{i,-}, p\rangle.
\end{equation}
Define $r\sps i:[2N]\to\bR$ such that $r\sps i = (f\sps i - p)\diamond (p - f\sps i)$. We compute $w\sps {i+1}$ using the multiplicative weights algorithm with step size $\alpha$:
\begin{align*}
\hat w\sps{i+1}(x) & \gets w\sps i(x)e^{-\alpha r\sps i(x)},\\
w\sps{i+1}(x) & \gets \hat w\sps{i+1}(x) / \sum_{x'\in [N]}\hat w\sps{i+1}(x')
\end{align*}
By \Cref{lm:mw}, for every $w\in \Delta_{2N}$,
\[
\alpha \langle w\sps i - w, r\sps i\rangle \le \dkl(w\|w\sps i) - \dkl(w\|w\sps{i+1}) + \alpha^2\|r\sps i\|_\infty^2.
\]
Summing up over $i = 1,\ldots,T$ and noting that $\dkl(w\|w\sps 1) \le \ln (2N)$, $\dkl(w\|w\sps {T+1}) \ge 0$ and $\|r\sps i\|_\infty \le 1$, we have
\[
\sum_{i\in [T]}\alpha\langle w\sps i - w, r\sps i\rangle \le \ln (2N) + \alpha^2T,
\]
which implies
\begin{equation}
\label{eq:fixed-weight-mw-2}
\frac 1T\sum_{i\in [T]}\langle w\sps i, r\sps i\rangle \le \frac{\ln (2N)}{\alpha T} + \alpha + \frac 1T \left\langle w, \sum_{i\in [T]} r\sps i\right\rangle.
\end{equation}
By the definition of $w\sps{i,+},w\sps{i,-}$ and $r\sps i$, 
\[
\langle w\sps i, r\sps i\rangle = \langle w\sps{i,+}, f\sps i - p\rangle + \langle w\sps{i,-}, f - q\sps i\rangle = \langle w\sps {i,+} - w\sps{i,-}, f\sps i - p\rangle \ge 0,
\]
where the last inequality follows from \eqref{eq:fixed-weight-mw-1}. Plugging this into \eqref{eq:fixed-weight-mw-2}, for every $w\in \Delta_{2N}$,
\begin{equation}
\label{eq:fixed-weight-mw-3}
0 \le \frac{\ln (2N)}{\alpha T} + \alpha + \frac 1T\left\langle w, \sum_{i\in [T]}r\sps i\right\rangle.
\end{equation}
By our definition $r\sps i = (f\sps i - p)\diamond (p - f\sps i)$ and \Cref{claim:infinity-norm}, there exists $w\in \Delta_{2N}$ such that $\langle w, \sum_{i\in [T]}r\sps i\rangle = - \|\sum_{i\in [T]}(f\sps i - p)\|_\infty$, so \eqref{eq:fixed-weight-mw-3} implies
\[
\left\|\frac 1T\sum_{i\in [T]}f\sps i - p\right\|_\infty \le \frac{\ln (2N)}{\alpha T} + \alpha.
\]
Setting $T = \lceil \frac{4\ln (2N)}{\varepsilon^2}\rceil$ and $\alpha = \sqrt{(\ln (2N))/T}$, we get $\|\frac 1T\sum_{i=1}^T f\sps i - p\|_\infty \le 2\sqrt{(\ln(2N))/T} \le \varepsilon$. The proof is completed by noting that the multiplicative weights algorithm allows us to compute $f\sps 1,\ldots,f\sps T$ in time $\tilde O(NT) = \tilde O(\varepsilon^{-2}N)$.
\end{proof}
The running time in \Cref{lm:fixed-weight-mw} can be further improved in the case where $N \gg 1/\varepsilon$:
\begin{lemma}
\label{lm:fixed-weight-mw-2}
Let $N > 1$ be a positive integer, and
let $p:[N] \to [0,1]$ be a predictor satisfying $\sum_{x\in [N]}p(x) = k$ for some integer $k$. For any $\varepsilon\in (0,1/2)$, there exists a distribution $\mu$ over \emph{binary} functions $f:[N] \to \{0,1\}$ such that
\begin{enumerate}
\item $\Pr_{M\sim\mu}[\sum_{x\in [N]}f(x) = k] = 1$, and
\item for every $x\in [N]$, $|\E_{f\sim\mu}[f(x)] - p(x)| \le \varepsilon$.
\end{enumerate}
Moreover, given $p$ as input, we can sample from the distribution $\mu$ in time $\tilde O(\varepsilon^{-2}r + N)$ where $r = \min\{1/\varepsilon,N\}$.
\end{lemma}
\begin{proof}[Proof]
We can assume w.l.o.g.\ that $N \ge 1/\varepsilon$ because of \Cref{lm:fixed-weight-mw}.
We partition $[0,1]$ into $u = \lceil 2/\varepsilon\rceil$ sub-intervals $\Gamma_1,\ldots,\Gamma_u$ with the width of each interval being at most $\varepsilon/2$. These sub-intervals partition $X:=[N]$ into $X_1,\ldots,X_u$ where $X_i:= \{x\in X:p(x)\in \Gamma_i\}$. Define $s(i):= \sum_{x\in X_i}p(x)$ and define $t(i):= s(i) - \lfloor s(i) \rfloor$ for every $i\in [u]$. We have $t(i)\in [0,1]$ and $\sum_i t(i)$ is an integer $k' = k - \sum_i\lfloor s(i) \rfloor$. Define $I_0:=\{i\in [u]:t(i) = 0\}$. By \Cref{lm:fixed-weight-mw}, we can find a distribution over $\hat t:[u]\setminus I_0\to\{0,1\}$ such that $|\E[\hat t(i)] - t(i)| \le \varepsilon / 2$ for every $i\in [u]\setminus I_0$ and $\sum_{i\in [u]\setminus I_0} \hat t(i) = k'$. 
Moreover, a random $\hat t$ can be sampled in time $\tilde O(\varepsilon^{-2}r)$ given $t$.
We define $\hat t(i) = 0$ for every $i\in I_0$. For every $i\in [u]$, we define $\hat s(i):= \lfloor s(i)\rfloor + \hat t(i)$. It is clear that $\hat s(i) \le \lceil s(i)\rceil \le |X_i|$ and
\begin{align}
{\sum}_{i\in [u]}\hat s(i) & = k, \quad \text{and }\label{eq:fixed-weight-mw-2-1}\\
|\E[\hat s(i)] - s(i)| & \le \varepsilon / 2. \label{eq:fixed-weight-mw-2-2}
\end{align}
Given $\hat t$, we compute $\hat s$ and
for each $i\in [u]$, we randomly pick $\hat s(i)$ individuals $x$ in $X_i$ and assign $f(x) = 1$, and we assign $f(x) = 0$ otherwise. This allows us to sample $f$ in time $\tilde O(N)$ given $\hat t$. By \eqref{eq:fixed-weight-mw-2-1} we have $\sum_xf(x) = k$. For each individual $x\in X_i$,
\[
\E_{f}[f(x)] = \E[\hat s(i)]/|X_i|.
\]
By \eqref{eq:fixed-weight-mw-2-2}, we have
\[
|\E_f[f(x)] - s(i)/|X_i|| \le \varepsilon / 2.
\]
Since the width of $\Gamma_i$ is at most $\varepsilon/2$, by the definition of $s(i)$ and $X_i$, we have $|s(i)/|X_i| - p(x)| \le \varepsilon/2$. Combining this with the inequality above, we have $|\E_f[f(x)] - p(x)| \le \varepsilon$, as desired.
\end{proof}
\subsubsection{Step 3: Indistinguishable Model with \texorpdfstring{$\polylog (N)$}{polylog(N)} Time Implementation}
Now we prove \Cref{lm:fixed-weight}, where $N = 2^n$ is exponentially large and we want our implementation to run in time $\poly(n,\varepsilon^{-1})$. Because of this efficiency requirement, we cannot directly apply \Cref{lm:fixed-weight-mw-2} which would only give us a $\poly(N,\varepsilon^{-1})$ time implementation. Our idea is to divide the domain $X = \{0,1\}^n$ into small groups each with size roughly $\Theta(\varepsilon^{-1})$ and apply \Cref{lm:fixed-weight-mw-2} separately on each group. This requires us to set a target support size for every group so that these sizes sum up to the total support size $k$ for the entire domain. We achieve this by creating a binary tree whose root corresponds to the entire domain and whose leaves correspond to the groups. We assign a target support size for every node of the tree from the root to the leaves. For every non-leaf node, we determine how to split its support size to its two children by estimating the sum of $p(x)$ in the sub-domains corresponding to the children. We make sure that the depth of the tree is at most $n$ so that we can efficiently reach any leaf from the root, which allows us to construct an efficient implementation.

Because of \Cref{lm:fixed-weight-mw-2}, it is w.l.o.g.\ to assume $2^n \ge 8/\varepsilon$ in \Cref{lm:fixed-weight}.
We start by describing an inefficient randomized algorithm that produces a function $f:X\to\{0,1\}$ given the predictor $p$. We will then choose the model $M$ as the distribution of the sample-access model $B$ induced by the random function $f$ and show that $M$ can be implemented efficiently.

The randomized algorithm producing the function $f$ operates as follows.
For every bit string $z\in \bigcup_{j=0}^n\{0,1\}^j$, we can view $z$ as a node in a binary tree with root being the empty string. 
Starting from the root, we assign an integer $k_z$ to certain nodes $z$ in the tree as follows. For the root node $z$, we set $k_z:= k$. Suppose we have assigned an integer $k_z$ to a node $z$ of the tree, we assign integers $k_{z\| 0}$ and $k_{z\|1}$ to the two children $z\|0$ and $z\|1$ of $z$ as follows. Define $z':= z\|0$ and let $X_{z'}$ be the set consisting of $x\in X$ such that $z'$ is a prefix of $x$. We draw $x_1,\ldots,x_m$ independently and uniformly from $X_{z'}$ for $m = \lceil 16n^2/\varepsilon^2\rceil$, and define $\ell_{z'}:= \frac{|X_{z'}|}{m}\sum_{i=1}^m p(x_i)$. We set $\hat k_{z\|0}:= \lfloor \ell_{z'}\rfloor\in [0,|X_{z'}|]$ and $\hat k_{z\|1}:= k_z - \hat k_{z\|0}$. We then define
\[
(k_{z\|0},k_{z\|1}) = \begin{cases}
(\hat k_{z\|0}, \hat k_{z\|1}), & \text{if } 0 \le \hat k_{z\|1} \le |X_{z\|1}|,\\
(k_z, 0), & \text{if } \hat k_{z\|1} < 0,\\
(k_z - |X_{z\|1}|, |X_{z\|1}|),& \text{if } \hat k_{z\|1} > |X_{z\|1}|.
\end{cases}
\]
The above procedure allows us to assign an integer $k_z$ to every node $z$, and we use these integers to construct a function $f:X\to\{0,1\}$. Specifically, choosing $n'$ to be the largest integer satisfying $2^{n - n'} \ge 8/\varepsilon$, for every $z\in \{0,1\}^{n'}$, we construct a function $f_z:X_z \to \{0,1\}$, and the function $f$ is then constructed by combining $f_z$ for all $z\in \{0,1\}^{n'}$.

To construct each $f_z$, we first construct a predictor $p':X_z\to[0,1]$ as follows. If $k_z \ge \sum_{x\in X_z}p(x)$, we increase each $p(x)$ for $x\in X_z$ to $p'(x)$ so that $\sum_{x\in X_z}p'(x) = k_z$. Similarly, we decrease $p(x)$ to $p'(x)$ if $k_z < \sum_{x\in X_z}p(x)$. Using \Cref{lm:fixed-weight-mw-2}, in time $\poly(|X_z|,\varepsilon^{-1}) = \poly(\varepsilon^{-1})$ we can randomly construct a function $f_z:X_z\to \{0,1\}$ such that for every $x\in X_z$,
\begin{equation}
\label{eq:fixed-weight-proof-2}
|\E f_z(x) - p'(x)| \le \varepsilon /2.
\end{equation}
where the expectation is in the randomness of $f_z$.

We analyze the above procedure of generating the function $f$ in the following lemma:
\begin{lemma}
\label{lm:fixed-weight-model}
In the procedure above, if we choose
$m \ge Cn^2/\varepsilon^2$ for a sufficiently large absolute constant $C > 0$, then
\begin{enumerate}
\item $\Pr[\sum_{x\in X}f(x) = k] = 1$;
\item $\sum_{x\in X}|\E[f(x)] - p(x)| \le |k - \sum_{x\in X}p(x)| + \varepsilon|X|$.
\end{enumerate}
Note that the probability and expectation are over the randomness in $f$.
\end{lemma}
We first prove two helper lemmas below before we prove \Cref{lm:fixed-weight-model}.

\begin{lemma}
\label{lm:alpha-beta}
Assume that $|k_z - \sum_{x\in X_z}p(x)| \le \alpha$ and $|\ell_{z\| 0} - \sum_{x\in X_{z\| 0}}p(x)| \le \beta$. Then
\[
\left|k_{z\| 0} - \sum_{x\in X_{z\| 0}}p(x)\right| + \left|k_{z\| 1} - \sum_{x\in X_{z\| 1}}p(x)\right| \le \alpha + 2\beta + 2.
\] 
\end{lemma}

\begin{proof}
By our choice of $\hat k_{z\| 0} = \lfloor \ell_{z\| 0}\rfloor$ and $\hat k_{z\|1} = k_z - \hat k_{z\| 0}$, we have
\[
\left|\hat k_{z\circ 0} - \sum_{x\in X_{z\| 0}}p(x)\right| \le \beta + 1,
\]
and
\[
\left|\hat k_{z\circ 1} - \sum_{x\in X_{z\| 1}}p(x)\right| \le \left|k_{z} - \sum_{x\in X_{z}}p(x)\right| + \left|\hat k_{z\circ 0} - \sum_{x\in X_{z\| 0}}p(x)\right| \le \alpha + \beta + 1.
\]
Summing up,
\[
\left|\hat k_{z\circ 0} - \sum_{x\in X_{z\| 0}}p(x)\right| + \left|\hat k_{z\circ 1} - \sum_{x\in X_{z\| 1}}p(x)\right| \le \alpha + 2\beta + 2.
\]
The lemma is proved by the observation that replacing $(\hat k_{z\|0},\hat k_{z\|1})$ with $(k_{z\|0},k_{z\|1})$  does not increase the LHS of the inequality above.
\end{proof}
\begin{lemma}
\label{lm:fixed-weight-level}
Let $n'\in\bZ_{\ge 0}$ satisfy $2^{n - n'} \ge \frac 8\varepsilon$. Then
\[
\sum_{z\in\{0,1\}^{n'}}\E\left|k_z - \sum_{x\in X_z}p(x)\right| \le \left|k - \sum_{x\in X}p(x)\right| + \varepsilon|X|/2.
\]
\end{lemma}
\begin{proof}
Consider an arbitrary $z\in \bigcup_{j=0}^{n - 1}\{0,1\}^j$ and define $z' = z\|0$.
By our definition of $\ell_{z'}:= \frac{|X_{z'}|}{m}\sum_{i=1}^mp(x_i)$ for $x_1,\ldots,x_m$ drawn i.i.d.\ from the uniform distribution over $X_{z'}$, we have $\E[\ell_{z'}] = \sum_{x\in X_{z'}}p(x)$ and $\var(\ell_{z'}) \le \frac{|X_{z'}|^2}{m}$. Therefore,
\[
\E\left|\ell_{z'} - \sum_{x\in X_{z'}}p(x)\right|\le \E\left[\left(\ell_{z'} - \sum_{x\in X_{z'}}p(x)\right)\right]^{1/2} = \var[\ell_{z'}]^{1/2} \le \frac{|X_{z'}|}{\sqrt m}.
\]
By our choice of $m \ge 16n^2/\varepsilon^2$, we have
\begin{equation}
\label{eq:fixed-weight-level-1}
\E\left|\ell_{z'} - \sum_{x\in X_{z'}}p(x)\right| \le \frac{\varepsilon |X_{z'}|}{4n} = \frac{\varepsilon |X_z|}{8n}.
\end{equation}
We prove the following inequality by induction on $j = 0,\ldots,n$:
\begin{equation}
\label{eq:fixed-weight-level-2}
\sum_{z\in \{0,1\}^j}\E\left|k_z - \sum_{x\in X_z}p(x)\right| \le \left|k - \sum_{x\in X}p(x)\right| + \left(\frac{\varepsilon j}{4n}+ 2^{-(n - j - 1)}\right) |X|.
\end{equation}
The inequality holds trivially for $j = 0$ because $k_z = k$ and $X_z = X$ when $z$ is the empty string. Now suppose the inequality holds for some $j = 0,\ldots,n - 1$, and we show that it also holds with $j$ replaced by $j+1$. For every $z\in \{0,1\}^j$, defining $z':= z\|0$, by \Cref{lm:alpha-beta} we have
\[
\left|k_{z\| 0} - \sum_{x\in X_{z\| 0}}p(x)\right| + \left|k_{z\| 1} - \sum_{x\in X_{z\| 1}}p(x)\right| \le \left|k_{z} - \sum_{x\in X_{z}}p(x)\right|  + 2\left|\ell_{z'} - \sum_{x\in X_{z'}}p(x)\right| + 2
\]
Summing up over $z\in \{0,1\}^j$ and taking expectation, we have
\begin{align*}
& \sum_{z\in \{0,1\}^{j+1}}\E\left|k_z - \sum_{x\in X_z}p(x)\right| \\
\le {} & \sum_{z\in \{0,1\}^j}\left(\E\left|k_z - \sum_{x\in X_z}p(x)\right|\right) + 2\sum_{z\in \{0,1\}^j}\E\left|\ell_{z'} - \sum_{x\in X_{z'}}p(x)\right| + 2^{j + 1}\\
\le {} & \left|k - \sum_{x\in X}p(x)\right| + \left(\frac{\varepsilon j}{4n}+ 2^{-(n - j - 1)}\right) |X| + \frac{\varepsilon}{4n}|X| + 2^{-(n - j - 1)}|X|\tag{by induction hypothesis and \eqref{eq:fixed-weight-level-1}}\\
\le {} & \left|k - \sum_{x\in X}p(x)\right| + \left(\frac{\varepsilon(j+1)}{4n} + 2^{-(n - j - 2)}\right)|X|.
\end{align*}
The lemma is proved by setting $j = n'$ in \eqref{eq:fixed-weight-level-2} and observing that $\frac{\varepsilon n'}{4n} \le \frac \varepsilon 4$ and $2^{-(n - n' - 1)} \le \frac \varepsilon 4$.
\end{proof}

\begin{proof}[Proof of \Cref{lm:fixed-weight-model}]
For every $z\in \{0,1\}^{n'}$, by the definition of $p'(x)$,
\begin{equation}
\label{eq:fixed-weight-proof-1}
\sum_{x\in X_z}\left|p'(x) - p(x)\right| = \left|k_z - \sum_{x\in X_z}p'(x)\right|.
\end{equation}
Combining \eqref{eq:fixed-weight-proof-1} and \eqref{eq:fixed-weight-proof-2},
\begin{align*}
\sum_{x\in X_z}|\E[f_z(x)|k_z] - p(x)| & \le \sum_{x\in X_z}|\E[f_z(x)|k_z] - p'(x)| + \sum_{x\in X_z}|p'(x) - p(x)|\\
& \le \varepsilon |X_z|/2 + \left|k_z - \sum_{x\in X_z}p'(x)\right|.
\end{align*}
Taking expectation over $k_z$,
\begin{align*}
\sum_{x\in X_z}|\E[f(x)] - p(x)| & = \sum_{x\in X_z}|\E[f_z(x)] - p(x)|\\
& = \sum_{x\in X_z}|\E_{k_z}\E[f_z(x)|k_z] - p(x)|\\
& \le \sum_{x\in X_z}\E_{k_z}|\E[f_z(x)|k_z] - p(x)|\\
& \le \varepsilon |X_z|/2 + \E\left|k_z - \sum_{x\in X_z}p(x)\right|.
\end{align*}
Summing up over $z\in \{0,1\}^{n'}$ and applying \Cref{lm:fixed-weight-level},
\[
\sum_{x\in X}|\E[f(x)] - p(x)| \le \varepsilon |X|/2 + \left|k - \sum_{x\in X}p(x)\right| + \varepsilon |X|/2 = \left|k - \sum_{x\in X}p(x)\right| + \varepsilon |X|.
\qedhere
\]
\end{proof}
\begin{proof}[Proof of \Cref{lm:fixed-weight}]
It suffices to prove that the model $M$ defined as the distribution of the sample-access object $B$ induced by the random function $f$ in \Cref{lm:fixed-weight-model} can be implemented efficiently, and the implementation itself can be computed efficiently from $p$.

To see that $M$ can be implemented efficiently, we first generate all the randomness needed in the procedure of creating $f$ by querying a random oracle. We then note that in order to compute $f(x)$, it suffices to compute $k_z$ for every prefix $z$ of $x$ and then compute $f_z$ for $z$ being the length-$n'$ prefix of $x$. This can be done efficiently in time $\poly(n,\varepsilon^{-1})$.
\end{proof}

\subsection{Learning Support-Access Binary Functions}\label{sec:support-func}
In the previous sections we showed how to learn and construct an implementation of an indistinguishable model for a function induced sample-access object $B^*$, i.e. there exists a function $f^*:X\rightarrow\{0,1\}$, and the distinguishers and learner both receives random samples of form $(x,f^*(x))$. In this section, we learn a support-access object induced by a function see \Cref{def:support-function}. For a binary function $f^*:X\rightarrow\{0,1\}$, the support-access object $B^*$ induced by $f^*$ outputs random samples from the set $\set{x | f^*(x)=1}$. We show how to construct an efficient implementation of a model that is indistinguishable from $B^*$.

\begin{description}
    \item[Distinguishers:] Let $\cD$ be a set of distinguishers, such that each $D\in\cD$ has an associated subset $S_D\subset X$. Given access to a sample $x$ from a object $B$,the distinguisher  accepts if $x\in S_D$. 
    
    We assume that for every $D\in\cD$, the set $S_D$ has known size and an efficient description, that on input $x$ answers if $x\in S_D$.
    \item[Auditor:] Let $\cD$ be defined as above. We say that an algorithm $\Lambda^{B^*,p}$ is an $(\varepsilon,\gamma,\delta)$ auditor for the collection of sets $\cS = \set{S_D|D\in\cD}$ if it has the following properties. Given access to a function-induced support access-object $B^*$ and query access to a predictor $p:X\rightarrow[0,1]$. If there exists $S\in\cS$ and $b\in\{-1,1\}$ such that 
    \[
    b\paren{\Pr_{x\sim B^*(\bot)}[x\in S]-\Pr_{x\sim p}[x\in S] } > \varepsilon. \]
    Then the auditor returns a set $S'$ such that with probability $1-\delta$,
    \begin{align}
        b\paren{\Pr_{x\sim B^*(\bot)}[x\in S']-\Pr_{x\sim p}[x\in S'] } > \gamma.\label{eq:audit-set}
    \end{align}
    Where  $x\sim p$ is the distribution generated from the predictor $p$, i.e. $\Pr_{x\sim p}[x=x']= \frac{p(x')}{\sum_{x''\in X}p(x'')}$.
\end{description}

Assuming an auditor for the collection of sets $\cS$, we show a learning algorithm for support-access function object $B^*$.
\begin{theorem}\label{thm:rand-one}
Let $\alpha\in[0,1]$ , and let $\cB$ be a collection of support access object induced by binary functions, such that $\forall B\in\cB, \E_{x\in X}[f_{B}(x)]=\alpha$. Let $\cD$ be a collection of distinguishers as described above. 

Let $\varepsilon,\gamma,\delta',\delta''$ be parameters such that $\delta'\leq c\delta\gamma^2\alpha^{-2}$ for a sufficiently small constant $c$.
Let $\Lambda$ be an $(\varepsilon,\gamma,\delta')$ auditor $\Lambda$ for $\cD$. Then there exists a $(2\varepsilon,\delta)$-learner $L$  for $\cB$ with respect to the distinguisher class  $\cD$. The learner $L$ runs in time $\poly(\gamma^{-1}\log(\delta^{-1})\alpha^{-1},W_1,W_2)$, where $W_1$ is the running time of the auditor $\Lambda$ and $W_2$ the circuit complexity of its output. The implementation $T$ that the learner outputs runs in time $\poly(\gamma^{-1}\log(\delta^{-1})\alpha^{-1},W_2)$
\end{theorem}
The learner $L$ in the theorem above has access to an auditor $\Lambda$ that can audit support-access objects. In section \Cref{sec:support-audit} we show that such auditor exists under similar conditions to an auditor for sample-access objects.

The learning algorithm is similar to the classic algorithm in \Cref{lm:predictor}, with an additional step of keeping the expected value of the model to be approximately $\alpha$. This extra step is required to show convergence for support-access objects.
We note that if the function is sparse, i.e. $\alpha$ is very small, then the algorithm is no longer efficient.

\paragraph{Learning algorithm:} We described the polynomial-time learner $L$. The learner maintains a list $F$ or tuples $(S,w)$, for $S\subseteq X$ and $w\in [0,1]$. 
Using the list, the learner calculates a predictor $p:X\rightarrow[0,1]$, by setting for each  $x\in X$, 
\begin{equation} 
\label{eq:list-to-p}
p(x) = \lcap\paren{ w\cdot\one(x\in S):(S,w)\in F},\end{equation}
where $\lcap$ caps the value $p(x)$ to the range $[0,1]$ after adding every list element, as defined in  \cref{eq:lcap}. For example, if the list is $(S_1,w_1),(S_2,w_2)$ then we have $p(x) = \numcap(\numcap(w_1\cdot\one(x\in S_1)) + w_2\cdot\one(x\in S_2))$.

For every $x\in X$, calculating $p(x)$ is done in $\abs F$ calls for the functions $\one(x\in S)$, for sets $S$ that the auditor returns.

The learner creates the list $F$ using the following algorithm.
\begin{enumerate}
    \item Initialization: set $F = (X,w_0=\alpha)$.
    \item Query $\Lambda^{B^*,p}$, where $p$ is the predictor generated from $F$. \label{item:main-loop}
    \begin{enumerate}
        \item If it does not return a set, output the current $F$.
        \item If the auditor returns a set $S\subseteq X$ and a sign $b\in\set{-1,1}$, set $F \leftarrow F\cup (S,b\gamma\alpha\frac{\abs{X}}{\abs{S}})$.\label{item:main-update}
    \end{enumerate}
    \item Choose $t=\log \abs X$ random elements from $X$, $I= \set{x_1,\ldots,x_t}$.
    
    If $\abs{\E_{x\in I}[p(x)]-\alpha}>\beta = \frac{1}{10}\gamma\alpha^2$,  add  $([N],\beta\sign\paren{\alpha-\E_{i\in I}[p_i]})$ to $F$. Repeat the current item until the condition is satisfied. Return to \Cref{item:main-loop}. \label{item:weight-update} \end{enumerate}
\paragraph{Implementation:} Given a list $F$, a random oracle implementation $T$ implements sampling $x\sim B(\bot)$ by rejection sampling. Denote by $R$ the random oracle, and assume that for every input the random oracle outputs a different random string.
\begin{enumerate}
   \item Choose a random $x\in X$.\label{item:init-samp}
    \item Use $R(x)$ to sample a bit $b$ such that $\Pr_{R(x)}[b=1]= p(x)$.
    \item If $b=1$, return $x$. If $b=0$, go back to \Cref{item:init-samp}.
\end{enumerate}
The expected number of loop iterations in the implementation is $\alpha$, since this is the expected value of $p$.
The random oracle is used to create an implement a consistent function object $B$. That is, in every time the implementation samples $x\in X$, the random oracle returns the same random string $R(x)$. This means that for every randomness, the implementation implements a support-access object for a specific function $f:X\rightarrow\set{0,1}$, as required.

\begin{proof}[Proof of \Cref{thm:rand-one}]
We start by assuming that the learning algorithm finished its running, and showing that in this case the output is indistinguishable from $B^*$ to all $D\in\cD$. 

The algorithm ends after \Cref{item:weight-update}, which means that with probability $(\beta \abs{ X})^{-1}$, we have 
\begin{align}
    \abs{ \sum_{x\in X}p(x)-\alpha \abs X} < 2\beta. \label{eq:final-weight}
\end{align}

Assuming \cref{eq:final-weight} holds, the implementation $T$ of our model is efficient and runs approximately $1/\alpha$ steps before stopping. It sample a random entry $x\sim p$, from $p$ such that the auditor does not return any set $S$, which means that for all $D\in\cD$,
\[  \abs{\Pr_{x\sim B^*(\bot)}[x\in S_D] - \frac{\sum_{x\in S_D}p(x)}{\sum_{x\in X}p(x)}}\leq\varepsilon. \]

Together this implies that after the algorithm ends, the model $M$ is such that
\[ \abs{\Pr_{x\in B^*(\bot)}[D(x)=1] - \Pr_{B\sim M,x\in B(\bot)}[D(x)=1]}\leq \varepsilon+2\beta\leq 2\varepsilon.\]
Therefore, if the learning algorithm ends it outputs an implementation for a model that is indistinguishable from $B^*$ for all $D\in\cD$.

We are left with proving that the learning algorithm is efficient, which also bounds the complexity of the implementation $T$ (as it depends on the size of $F$). We prove this using a potential function:\[ \varphi(B^*,p) = \sum_{x\in X}\paren{f^*(x) - p(x)}^2.\]

Since both $p$ and $f^*$ are in $[0,1]$, the value of $\varphi$ is bounded by $\abs X$. 

By definition, with probability $1-\delta'$ the auditor returns a set $S$ that satisfies \cref{eq:audit-set}. In this case, from \Cref{claim:disting-update}, the value of $\varphi$ is reduced by $\Omega(\gamma^2\alpha^2 \abs X)$ on each iteration. In addition, from \Cref{claim:w0-update}, the value of $\varphi$ is reduced by $\Omega(\gamma^2\alpha^4 \abs X)$ on each iteration of \Cref{item:weight-update}. Together, we get that the total number of weight updates (of both types) are at most $O(\gamma^{-2}\alpha^{-4})$.
\end{proof}

\begin{claim}\label{claim:disting-update}
Suppose that on \Cref{item:main-update} of the algorithm, the learning algorithm list is updated from $F$ to $F\cup (S,w)$, such that $(S,w)$ satisfies \cref{eq:audit-set}.
Let $p$ be the predictor described by $F$, and $p'$ be the predictor described by $F\cup (S,w)$. Assume that $\abs{\E_{x\in X}[p(x)]-\alpha}\leq 2\delta$, then $\varphi(B^*,p')\leq \varphi(B^*,p) - \frac{1}{2}\gamma^2\alpha^2 \abs X$.
\end{claim}
\begin{proof}
According to our assumption, we have
 \[ \abs{\Pr_{x\sim B^*(\bot)}[x\in S] - \frac{\sum_{x\in S}p(x)}{\sum_{x\in X}p(x)}}   > \gamma. \]
The previous inequality can be written also as
\begin{align}
    \abs{\frac{\sum_{x\in S}f^*(x)}{\sum_{x\in X}f^*(x)}- \frac{\sum_{x\in S}p(x)}{\sum_{x\in X}p(x)}} > \gamma.\label{eq:disting}
\end{align}
We divide into two cases, of $b=1$ and $b=-1$. The two cases are nearly identical and we write them both for completeness.
\paragraph{If $b=1$:} Then \cref{eq:disting} can be written as:
\begin{align}\label{eq:up-bound}
\gamma < {} & \frac{\sum_{x\in S}f^*(x)}{\sum_{x\in X}f^*(x)}- \frac{\sum_{x\in S}p(x)}{\sum_{x\in X}p(x)}\\
\leq {} & \frac{1}{\alpha \abs X}\sum_{x\in S}f^*(x) - \frac{1}{\abs X\paren{\alpha+2\delta}}\sum_{x\in S}p(x)\nonumber\\
\leq {} & \frac{1}{\alpha \abs X}\sum_{x\in S}\paren{f^*(x) - \paren{1-\frac{2\delta}{\alpha+2\delta}}p(x)}
\nonumber
\end{align}
Where we use the fact that $\E_{x\in X}[p(x)]\leq\alpha+\beta$.
\begin{align*}
    \varphi(B^*,p) - \varphi(B^*,p') = {} &\sum_{x\in S}\paren{2f^*(x) w -2w p(x) - w^2}
    \\= {} &2w\sum_{x\in S}\paren{f^*(x)-p(x)}-\abs{S}w^2.
    \\\geq {} &2 w \gamma\alpha N  - 2w\frac{2\delta}{\alpha+2\delta}\abs{S}- w^2\abs{S} ,\\
    = {} &2\frac{\gamma^2\alpha^2 \abs{X}^2}{\abs{S}} - 4\abs{X}\gamma\delta - \varepsilon^2 \alpha^2\frac{\abs{X}^2}{\abs{S}}\\
    \geq {} & \frac{1}{2}\frac{\gamma^2\alpha^2\abs{X}^2}{\abs{S}}\geq \frac{1}{2}\gamma^2\alpha^2\abs X.
\end{align*}
Where we use the fact that $\beta = \frac{1}{8}\gamma\alpha^2$.

The capping of $p(x)$ to $[0,1]$ can only reduce the value of the potential function, because for every $x\in X$ we have $f^*(x)\in\set{0,1}$.

\paragraph{If $b=-1$:} This case is nearly identical to the $b=1$ case.
In this case, \cref{eq:disting} can be written as
\begin{align*}
\gamma < {} &  \frac{\sum_{x\in S}p(x)}{\sum_{x\in X}p(x)}- \frac{\sum_{x\in S}f^*(x)}{\sum_{x\in X}f^*(x)}\\
\leq {} &  \frac{1}{\abs{X}(\alpha-2\beta)}\sum_{x\in S}p(x)-\frac{1}{\alpha \abs X}\sum_{x\in S}f^*(x)\\
\leq {} & \frac{1}{\alpha \abs X}\paren{\sum_{x\in S}(p(x)-f^*(x)) + \frac{2\beta}{\alpha-2\beta}\abs{S}}.
\end{align*}
If the algorithm did not perform any capping:
\begin{align*}
    \varphi(B^*,p) - \varphi(B^*,p')= {} &\sum_{x\in S}(2f^*(x) w -2w p(x) - w^2) \\
    \geq {} & \frac{\gamma^2\alpha^2\abs{X}^2}{\abs{S}} - 8\gamma\beta \abs{X}\geq \frac{1}{5}\gamma^2\alpha^2\abs X.
\end{align*}
As before, capping $p(x)$ to $[0,1]$ can only reduce the value of $\varphi$.
\end{proof}

\begin{claim}\label{claim:w0-update}
Suppose that in the step described by \Cref{item:weight-update}, the algorithm makes an update to the list $F$. 
Let $p$ be the predictor before the update (as defined in \eqref{eq:list-to-p}), and $p'$ be the predictor after the update. Then with probability $1-O(\abs{X}^{-1})$, $\varphi(B^*,p')\leq \varphi(B^*,p) -  \frac{9}{10}\beta^2\abs{X}$. 
\end{claim}
\begin{proof}
This proof is the classic learning proof.
Using the Chernoff bound, with probability $O(\abs{X}^{-1})$ we have that
\[ 
\abs{\sum_{x\in I}\frac{1}{t}p(x) - \sum_{x\in X}\frac{1}{N}p(x)} \leq \frac{1}{100}\beta.\]
Assuming the above holds, then 
\[ \abs{\sum_{x\in X}f^*(x) - \sum_{x\in X}p(x)} > \frac{99}{100}\beta \abs X.\]

Assume without loss of generality that $\sum_{x\in X}f^*(x)\geq\sum_{x\in X}p(x)$, then we add $([N],\beta)$ to the list.
We have 
\begin{align*}
    \varphi(B^*,p) - \varphi(B^*,p') = {} & \sum_{x\in X}\left((f^*(x) - p(x))^2 - (f^*(x) - p'(x))^2\right)
    \\
    = {} &\sum_{x\in X}\left(-2f^*(x)(p(x)-p'(x)) + p(x)^2 - p'(x)^2\right).
\end{align*}
Assume for now that no capping is done for $p'$. In this case, $p'(x) = p(x) + \beta$ and thus
\begin{align*}
        \varphi(B^*,p) - \varphi(B^*,p')\geq {} & \sum_{x\in X}\left(2f^*(x) \beta - 2\beta p(x) -\beta^2\right)\\
        \geq {} & 2\cdot \frac{99}{100}\beta^2 \abs X - \beta^2 \abs X = \frac{98}{100}\beta^2\abs X.
\end{align*}
Capping $p'$ can only reduce the value of $\varphi(B^*,p')$, since $f^*(x)\in[0,1]$.
\end{proof}

\subsubsection{Auditor for Support-Access Objects}\label{sec:support-audit}
The auditor for the sample-access objects can be seen as a weak agnostic learner of the concept class $\cG$ representing the class of distinguishers $\cD$ under the uniform distribution. In the case of support-access objects, the auditor does not get labeled samples $(x,f^*(x))$, but rather a random $x$ such that $f^*(x)=1$.

We show that for a large family of learning algorithms, statistical query algorithms, learning from support-access objects is not much more difficult than learning from sample-access objects.
\begin{definition}[Statistical Query Oracle]
    A statistical query oracle over a distribution $P:X\times\{0,1\}\rightarrow[0,1]$ is a function receiving as an input a function $\phi:X\times\{0,1\}\rightarrow \{0,1\}$ and error parameter $\varepsilon$ and outputting a value $v\in[0,1]$ such that  $\abs{v-\E_{(x,y)\sim P}[\phi(x,y)]}\leq \varepsilon$.
\end{definition}
\begin{definition}[Statistical Query Algorithm]
    An algorithm $\cA$ is a statistical query algorithm over a distribution $P:X\times\{0,1\}\rightarrow[0,1]$ if it only access to the object is using a statistical query oracle over the distribution $P$.
\end{definition}
\begin{claim}
    If there exists a statistical query algorithm $\cA$ learning a hypothesis $h$ over the distribution $P$ created from sample-access object, i.e. that outputs $(x,f^*(x))$ for a uniform $x\in X$, then there exists an algorithm $\cA'$ learning $h$ from the distribution $P'$ created from support-access object, i.e. that outputs a uniform $(x,1)$ for a random $x\in X$ such that $f^*(x)=1$. The algorithm $\cA'$ also requires $m=\abs{\supp(f^*)}=\sum_{x\in X}f^*(x)$ as an input. 

    Furthermore, $\cA'$ has the same number of queries to the statistical oracle, has running time similar to the running time of $\cA$ up to $\polylog \abs X$ factors, and has the same success probability up to an additive error of $1/\log\abs{X}$. 
\end{claim}
\begin{proof}
    The algorithm $\cA'$ simulates the algorithm $\cA$, and simulates the the statistical query oracle over the distribution $P$ using the oracle for $P'$. Suppose the algorithm $\cA$ queries the statistical query oracle the function $\phi:X\times\{0,1\}\rightarrow\{0,1\}$. Since the second variable is binary we can write,
    \begin{align}
        &\phi(x,y) = \varphi_1(x)\cdot y + \varphi_2(x)(1-y),\\
        &\E_{x\in X}[\phi(x,f^*(x))] = \E_{x\in X}[(\varphi_1(x)-\varphi_2(x)) f^*(x) ] + \E_{x\in X}[\varphi_2(x)]. 
    \end{align}
    
    The value $ \E_{x\in X}[\varphi_2(x)]$ does not depend on the function $f^*$ and can be approximated without any queries to $f^*$ or to the oracle. Given that the function $\varphi_2(x)$ is bounded, $\cA'$ can approximate $ E_{x\in X}[\varphi_2(x)]$ up to an additive error of $\varepsilon/2$ in $\polylog \abs{X}$ runtime, with success probability $1-\abs{X}^{-2}$. We denote this approximation by $v_1$. 

    For the second element, $\E_{x\in X}[(\varphi_1(x)-\varphi_2(x)) f^*(x) ]$, we have 
    \begin{align*}
         \E_{x\in X}[(\varphi_1(x)-\varphi_2(x)) f^*(x) ] & = \Pr_{x\in X}[f^*(x)=1]\cdot \E_{x\in X}[(\varphi_1(x)-\varphi_2(x))| f^*(x)=1 ] \\
         & =\frac{m}{\abs{X}}\E_{x\in X}[(\varphi_1(x)-\varphi_2(x))| f^*(x)=1 ].
    \end{align*}
    Notice that the last term is the expected value of $\varphi_1(x)-\varphi_2(x)$ under the support-access object distribution, $P'$.
    Therefore, the algorithm $\cA'$ then queries the statistical query oracle over distribution $P'$ with function $\phi'(x,y)= \varphi_1(x)-\varphi_2(x)$ and error parameter $\varepsilon/2$. Denote the return value to be $v_2$. Then, $\cA'$ to simulates $\cA$ assuming the return value of the statistical query oracle  to $\cA$ is $v_1+\frac{m}{\abs{X}}v_2$. Since the error in  $v_1,v_2$ is bounded by $\varepsilon/2$, the return value is a valid return value from the statistical oracle query with probability $1-\abs{X}^{-2}$. 
    
    Summing over all of the queries of $\cA$ to the statistical query oracle, the probability that $\cA'$ fails is the same as $\cA$ with the additional factor smaller than $1/\log\abs{X}$ (which is the error in the approximation of $ \E_{x\in X}[\varphi_2(x)]$).
\end{proof}
Therefore, if there exists an auditor in the standard sample-access object, that is a statistical query algorithm, then there is also an auditor for the support-access object. 

We further prove the existence of an auditor for support-access objects for specific set of distinguishers.
\begin{claim}
    Let $\cD$ be a set of $t$ distinguishers for the support-access object, as defined on \Cref{sec:support-func}. Then there exists an auditor $\Lambda$ as defined in \Cref{sec:support-func} with running time $\Theta(t)$ and query complexity $O(\log(t)\polylog(\abs{X})1/(\delta(\varepsilon-\gamma)))$.
\end{claim}
\begin{proof}
    The auditor samples $k = c\cdot\log(t)\polylog (\abs{X})$ uniform samples $x_1,\ldots,x_k\sim B^*(\bot)$. For every distinguisher $D\in\cD$ with an associated set $S_D$ the auditor calculates 
    \[ v_D = \Pr_{x\in \{x_1,\ldots, x_k\}}[x\in S_D ] = \frac{\abs{\{ x_1,\ldots,x_k\}\cap S_D}}{k}. \]

    If there exists $D\in\cD$ such that $\abs{v_D  - \Pr_{x\sim p}[x\in S_D]}>\frac{\varepsilon+\gamma}{2}$, the auditor returns $S_D$.

    For every set $S_D\subseteq X$, by Chernoff bound the probability over $x_1,\ldots,x_k$ that 
    \[ \abs{\Pr_{x\in \{x_1,\ldots, x_k\}}[x\in S_D ] - \Pr_{x\in B^*(\bot)}[x\in S_D ]}\geq \frac{\varepsilon-\gamma}{2}\]
    is at most $e^{-k(\varepsilon-\gamma)^2/10}$. Therefore, we choose the constant $c$ such that even when summing over the $t$ distinguishers, the error is smaller than $\delta$.
\end{proof}

We can look at the value $v_D$ calculated by the algorithm above as an estimator for $\Pr_{x\in B^*(\bot)}[x\in S_D ]$. The argument above implies that $v_D$ uniformly converge to $\Pr_{x\in B^*(\bot)}[x\in S_D ]$.
\begin{definition}[Uniform Convergence]
 An estimator $v$ uniformly converges to $P$ under distribution $B$ over the hypothesis class $\cD$ if there exists a function $k:(0,1)\times(0,1)\rightarrow\mathbb{N}$ such that for every $\varepsilon,\delta >0$, if $k>k(\varepsilon,\delta)$ then
 \[  \Pr_{x_1,\ldots,x_k\sim B}[\exists D\in\cD \text{ s.t.\ }  \abs{v_D(x_1,\ldots,x_k) - P_{D,x\sim B}}\geq\varepsilon] <\delta. \]
\end{definition}
The uniform convergence implies that there exists an algorithm for weak agnostic learning of the class $\cD$ with sample complexity $\poly(1/\varepsilon,\log(1/\delta), \mathsf{VCdim}(\cD))$ \cite{doi:10.1137/1116025,MR1072253,MR1088804}. 

\subsection{Learning Bit-String Functions}\label{sec:bit-string-func}
In this section we are interested in learning a function $f:\set{0,1}^n\rightarrow \set{0,1}^n$. This is a harder than learning a binary function, because the range of the function is very large. In this section, we only learn a indistinguishable model with respect to a very limited set of distinguishers, with a product structure. In this setting, the sampling distribution is a pair $(x,f(x))$ for a random input $x$.

\begin{description}
\item[Distinguishers:] Let $\cD$, such that each distinguisher $D\in \cD$ has an set $S_D\subset \set{0,1}^n$ and a coordinate $j\in[n]$. The distinguisher $D$ accept a sample $(x,f(x))$ if $x\in S$ and $f(x)_j=1$.
\item[Auditor:] Let $\cD$ be defined as above. We say that an algorithm $\Lambda^{B^*,p}$ is an $(\varepsilon,\gamma,\delta)$ auditor for the collection of sets $\cS$ if it has the following properties. Given access to a function-induced support access-object $B^*$ and query access to a set of $n$ predictors $p_1,\ldots p_n$, such that $p_j:\set{0,1}^n
\rightarrow[0,1]$. If there exists $S\in\cS$ and $j\in[n]$ such that 
    \[
    b\paren{\Pr_{x}[x\in S,f(x)_j=1]-\E_{x}[ p_j(x)\cdot\one(x\in S)] } > \varepsilon. \]
    Then the auditor returns a set $S'\subseteq\set{0,1}^n$ and $j\in[n]$ such that with probability $1-\delta$,
    \[
    b\paren{\Pr_{x}[x\in S',f(x)_j=1]-\E_{x}[ p_j(x)\cdot\one(x\in S')] }  > \gamma. \].
\end{description}

\begin{theorem}\label{thm:func}
Let $\cB$ be a collection of support access object induced by  functions $f:\set{0,1}^n\rightarrow \set{0,1}^n$. Let $\cD$ be a collection of distinguishers as described above. 

Let $\varepsilon,\gamma,\delta',\delta''$ be parameters such that $\delta'\leq c\delta\gamma^2 n^{-1}$ for a sufficiently small constant $c$.
Let $\Lambda$ be an $(\varepsilon,\gamma,\delta')$ auditor for $\cD$. Then there exists a $(2\varepsilon,\delta)$-learner $L$ to $\cB$ with respect to the distinguisher class  $\cD$ and the learner $L$ runs in time $\poly(\gamma^{-1}\log(\delta^{-1})\alpha^{-1},W_1,W_2)$, where $W_1$ is the running time of the auditor $\Lambda$ and $W_2$ the circuit complexity of its output. The implementation $T$ that the learner outputs runs in time $\poly(\gamma^{-1}\log(\delta^{-1})\alpha^{-1},W_2)$
\end{theorem}

We describe the learner $L$ and implementation $T$.
\begin{description}
\item[Implementation:] Given a list $F$ of tuples $(S,j,w)$, the implementation algorithm $T$ with oracle access to a random oracle $R$ generates a model $\hat B$. Given $x\in\set{0,1}^n$ the algorithm $T$ outputs $v$ by:
\begin{enumerate}
    \item For every $j\in[n]$ calculate $p_j(x)\in [0,1]$ by: 
    \[p_j(x) = \lcap\paren{w\cdot\one(x\in S): (S,j,w)\in F}. \]
    \item Choose $v\in\set{0,1}^n$ by setting  for every $j\in[j]$, $v_j=1$ with probability $p_j(x)$ independently using the random oracle $R(x)$.
    \end{enumerate}

\item[Learning Algorithm:] The learning algorithm $L$.
\begin{enumerate}
    \item Initialization: set $F = \set{(\set{0,1}^n,j,1/2)|j\in[n]}$.
    \item Query the auditor with the current model and $B^*$. If the auditor returns $(S,j,b)$, add $(S,j,b\cdot\gamma)$ to $F$ and repeat.\label{item:prod-update}
\end{enumerate}
\end{description}
Notice that the learner is identical to the classic boosting algorithm, where we apply it individually on each of the coordinated of $f(x)$. This is the reason that we can only be indistinguishable with respect to a set of distinguishers $D$ that only check if a single output coordinate in $f(x)$ equals $1$.

\begin{proof}[Proof of \Cref{thm:func}]
Similar to the previous proofs, if the auditor does not output any set then the model is $\varepsilon$-indistinguishable for all $D\in\cD$.

We prove the correctness using a potential function, which sums over the error of each coordinate $j\in[n]$. 
Let \[\varphi(B^*,p_1,\ldots,p_n) =\sum_{j\in[n]}\sum_{x\in\set{0,1}^n}\paren{\one(f(x)_j=1) - p_j(x)}^2. \]
Every $x$ appears in the sum $n$ times, therefor the initial value of $\varphi$ is bounded by $n\cdot 2^n$.

Let $p_1,\ldots,p_n$ be the predictors before the execution of \Cref{item:prod-update}, and $p'_1,\ldots, p_n'$ afterwards. We claim that with probability $1-\delta'$ we have that
$\varphi(B^*,p_1,\ldots,p_n) - \varphi(B^*,p_1',\ldots,p_n')\geq\gamma^2 2^n$.

The auditor returned $(S,i,b)$, then 
\begin{align}
    \abs{ \frac{1}{2^n}\sum_{x\in S}(\one(f(x)_j=1) - p_j(x))}\geq \gamma.\label{eq:prod-adv}
\end{align}
Assume without loss of generality that $b=1$.  We note that the difference  between $p_1,\ldots,p_n$ and $p_1',\ldots,p_n'$ is only between $p_i,p_i'$ and only for $x\in S$. 
\begin{align*}
    \varphi(B^*,p_1,\ldots,p_n) - \varphi(B^*,p'_1,\ldots,p_n') = {} & \sum_{x\in S}\left( \paren{\one(f(x)_i=1) - p_i(x)}^2 - \paren{\one(f(x)_i=1) - p'_i(x)}^2\right)  \\= {} &\sum_{x\in S}\left(-2\cdot\one(f(x)_i=1)(p_i(x)- p'_i(x))+  p^2_i(x) - {p'_i(x)}^2\right).
\end{align*}
Assuming that the algorithm does not do any capping,
\begin{align*}
    \varphi(B^*,p_1,\ldots,p_n) - \varphi(B^*,p'_1,\ldots,p_n') = {} & \sum_{x\in S}\left(2\gamma\cdot\one(f(x)_i=1) - 2\gamma p_i(x) -\gamma^2\right)\\\geq {} &
    2\gamma^2 2^n -\gamma^22^n\\
    = {} & \gamma^2 2^n,
\end{align*}
where we use \cref{eq:prod-adv} in the last inequality.

The values of $\one(f(x)_j=1)$ are in $\set{0,1}$, therefore capping the value of $p_i$ to $[0,1]$ can only lower the value of the potential function.
\end{proof}

\section{Learning Exponential-Size Graphs}

\subsection{Learning Dense Graphs}\label{sec:dense}
The most basic setting for graphs is the dense model, where the graph-induced sample-access object $B$ induced by a graph $G=([N],E)$ can be thought of as getting a random entry $((u,v),b)$ for $u,v\in[N],b\in\set{0,1}$ from the adjacency matrix of $G$ (see \Cref{def:sample-graph}). In this setting, we can think of the adjacency matrix of the graph as a function, where the graph imposes some extra structure on the distinguishers. Therefore, some of the results from \Cref{sec:func-learning} follow directly.

\paragraph{Distinguishers:} Let $\cD$ be a collection of distinguishers. Each distinguisher has two sets of vertices $U_D,V_D\subset [N]$. When getting a random entry $((u,v),b)$ from the graph-induced sample-access object $B^*$, it accepts if $u\in U_D$, $v\in V_D$ and $b=1$.

\paragraph{Auditors:} An algorithm $\Lambda$ is an $(\varepsilon,\gamma,\delta)$-auditor for a collection of tuples $\cS$ containing pairs $(U,V)\subset [N]\times[N]$ if it satisfies the following.  The auditor received query access to a predictor $p:[N]\times[N]\rightarrow[0,1]$ and to a graph-induced sample-access object $B^*$ induced by a graph $G^* = ([N],E^*)$. If there exists $(U,V)\in\cS$ such that
\[ b\cdot\paren{\Pr_{(u,v)\in [N]\times[N]}[u\in U,v\in V,(u,v)\in E^*] - \E_{(u,v)\in [N]\times[N]}[\one(u\in U,v\in V)p(u,v)]}\geq\varepsilon. \]
Then with probability $(1-\delta)$, it outputs $U',V'$ such that
\[ b\cdot\paren{\Pr_{(u,v)\in [N]\times[N]}[u\in U',v\in V',(u,v)\in E^*] - \E_{(u,v)\in [N]\times[N]}[\one(u\in U',v\in V')p(u,v)]}\geq\gamma. \]

From \Cref{thm:sample-function}, which uses a multicalibrated predictor from \cite{TrevisanTV09,hebert2018multicalibration,kim2019multiaccuracy}, we get that we can learn an exponential graph model assuming we have an auditor for the class of distinguishers. This is done by simply applying \Cref{thm:sample-function} on the sample-access object for a graph, when treating its adjacency matrix as a binary function, i.e. apply it on $f^*:[N]\times[N]\rightarrow\set{0,1}$ such that $f^*(u,v)=1\iff (u,v)\in E^*$. 
\begin{corollary}[Corollary of \Cref{thm:sample-function}]
\label{thm:sample-graph}
Let the distinguisher class $\cD$ be defined above. Let $\varepsilon,\gamma,\delta,\delta'> 0$ be parameters satisfying $\delta' \le c\delta\gamma^2$ for a sufficiently small absolute constant $c > 0$. Let $\cB$ be the class of sample-access objects induced by graphs over vertex set $[N]$. Let $\Lambda$ be an $(\varepsilon,\gamma,\delta')$-auditor for $\cS = \set{(U_D,V_D)|D\in\cD}$. Then there exists an $(\varepsilon,\delta)$-learner $L$ for $\cB$ w.r.t.\ $\cD$.
Moreover, if the auditor $\Lambda$ runs in time at most $W_1$ and always outputs a function with circuit size at most $W_2$, then the learner $L$ runs in time $\poly(\gamma^{-1},\log(\delta^{-1}),W_1)$ and always outputs implementations with circuit complexity $O(\gamma^{-2}W_2)$.
\end{corollary}

The same holds also for learning a graph with a fixed number of edges. We can apply \Cref{thm:fixed-weight} on the sample-access object representing the adjacency matrix of a graph with $m$ edges (when $m$ is known).  \Cref{thm:fixed-weight} implies that we can learn an indistinguishable model in the setting of \Cref{thm:sample-graph}, while also requiring that the model only contains graphs with $m$ edges. Note though that this can be done when using the adjacency matrix of the graph as a function, so the result is meaningful only for dense graphs, i.e. when $m=\Theta(N^2)$.

The two theorems also apply when learning a model of a directed graph, and all of the previous results still holds. This is done simply by having the adjacency matrix of a directed graph, i.e. when $f(u,v)$ might differ from $f(v,u)$. In addition, for a directed graph, it is possible to apply the construction used in \Cref{sec:fixed-weight} to have a truthfulness requirement of a fixed out degree $d$. That is, when given a sample-access object induced by a directed graph $G$ with uniform out-degree $d$, we output a model such that for all $B\sim M$, the graph $G_B$ has out degree $d$ for all vertices.
This is done by learning a predictor as previously, and using the construction from \Cref{sec:fixed-weight} for the sampling algorithm. In the sampling algorithm, instead of having a directed edge $(u,v)$ with probability $p(u,v)$ independently, we apply the algorithm from \Cref{sec:fixed-weight} on the function representing the row of $u$ in the adjacency matrix, with the value of the function being $d$. We sample a directed edge by following the algorithm from \Cref{sec:fixed-weight}. 

\subsection{Learning Sparse Graphs Without Dense Subgraphs}
For a sparse graph, a sample-access graph object is not useful. This is because a random entry in the adjacency matrix of a sparse graph is nearly always $0$, so any sparse graph is indistinguishable from the empty graph for sample-access objects. For sparse graphs, we study graph-induced support-access objects (\Cref{def:support-graph}). A support-access object $B^*$ induced by a sparse graph $G = ([N],E)$, $B(\bot)$ outputs a random edge in the graph $(u,v)\in E$.

In \Cref{sec:support-func} we studied function-induced support access objects, i.e. an object of a function $f$ that returns a random input $x$ such that $f(x)=1$. For a graph $G=([N],E)$, if we view its adjacency matrix as a function $f:[N]\times[N]\rightarrow\set{0,1}$, then a support-access function object $B$ of $f$ returns a random positive entry of $f$, i.e. a random edge $(u,v)\in E$. Therefore, we can apply \Cref{thm:rand-one} for support-access graph objects.  Unfortunately, \Cref{thm:rand-one} only applies for functions $f$ with a constant expected value, i.e. when $f$ represents the adjacency matrices of a dense graph (for $f$ representing a sparse graph, the learning and implementation are inefficient).

In this section we show how to create \emph{a dense model for a sparse graph}, as long as the sparse graph does not have a subgraph which is too dense. We do so by using the strong regularity lemma for sparse graphs \cite{KR03,Scott2011}, which implies that for a sparse graph $G$ with no dense subgraphs, there exists a dense graph that is indistinguishable from it. Then we reduce the problem to finding a dense model for a dense graph, and apply \Cref{thm:rand-one}.

\paragraph{Distinguishers:} Let $\cD$ be a collection of distinguishers. Each distinguisher has two sets of vertices $U_D,V_D\subset [N]$. A distinguisher $D$ on input $(u,v)$ accepts if $u\in U_D$ and $v\in V_D$.

\paragraph{Auditors:} An algorithm $\Lambda$ is an $(\varepsilon,\varepsilon',\delta)$-auditor for a collection pairs of sets $\cS$, such that $(U,V)\in\cS$, $U,V\subset [N]$ if it satisfies the following. The auditor received query access to a predictor $p:[N]\times[N]\rightarrow[0,1]$ and access to a graph-induced support-access object $B^*$ representing a graph $G^* = ([N],E^*)$. If there exists a pair of sets $(U,V)\subset \cS$ and a bit $b$ such that
\[ b\cdot\paren{\Pr_{(u,v)\sim B^*(\bot)}[u\in U,v\in V] - \Pr_{(u,v)\sim p}[u\in U,v\in V]}\geq\varepsilon, \]
Then $\Lambda$ outputs sets $U',V'\subset [N]$ such that 
\[ b\cdot\paren{\Pr_{(u,v)\sim B^*(\bot)}[u\in U',v\in V'] - \Pr_{(u,v)\sim p}[u\in U',v\in V']}\geq\varepsilon'. \]
The distribution $(u,v)\sim p$ is defined by the predictor $p$,  i.e. for all $u,v\in [N]$ we have 
\[ \Pr_{(u',v')\sim p}[u'=u,v'=v] = \frac{p(u,v)}{\sum_{u'',v''\in [N]}p(u'',v'')}. \]

\paragraph{Graph Notations and Definitions}
For a graph $G=([N],E)$ and $U,V\subset [N]$, we define $E_G(U,V) = \set{(u,v)\in E | u\in U,v\in V}$ to be the set of edges between $U,V$ in $G$. We denote by $\rho_G(U,V)$ the edge density between $U,V$ in $G$, $\rho_G(U,V) = \frac{\abs{E_G(U,V)}}{\abs{U}\abs{V}}$. We denote by $\rho_G = \rho_G([N],[N])$ the edge density of the graph.

We use the definition of upper-uniform graphs from \cite{KR03,Scott2011} to define graphs without dense subgraph. We need a slightly stronger definition than what is used in \cite{KR03,Scott2011}, with an additional requirement also for $U,V$ smaller than $\eta N$.
\begin{definition}[Upper-uniform graphs]
    A graph $G=([N],E)$ is $(\eta,\gamma)$-upper uniform, if for every two disjoint sets $U,V\subset [N]$, with $\min\set{\abs{U},\abs{V}}\geq \eta N$ we have that
    \[ \rho_G(U,V)\leq \gamma\rho_G, \]
    and for $U,V$ such that $\min\set{\abs{U},\abs{V}}< \eta N$, 
    \[ \abs{E(U,V)}\leq \gamma \eta\rho_G N^2. \]
\end{definition}
We remark that with high probability, a random sparse graph is upper-uniform for constants $\eta,\gamma$.

In this section we use the regularity lemma for sparse graphs from \cite{KR03,Scott2011}. 
\begin{definition}[Regular sets]
    For graph $G=([N],E)$, the vertex sets $V,U\subset [N]$ are $(\delta)$-regular if for all $U'\subset U,V'\subset V$ with $\abs{U'}\geq\delta\abs{U},\abs{V'}\geq\delta \abs{V}$ we have 
    \[ \abs{\rho_G(U',V') - \rho_G(U,V)}\leq \delta\rho_G .\]
\end{definition}

\begin{theorem}[Regularity lemma for sparse upper uniform graphs \cite{KR03}]\label{lem:regularity}
    For every $\delta,\gamma>0$ there exists $(\eta,m) = \varphi(\delta,\gamma), \eta>0,m\in\mathbb{N}$  such that every $(\eta,\gamma)$-upper uniform graph $G$ has a partition $U_0,U_1,\ldots,U_k$ with $k\leq m$ such that $\abs{V_0}\leq \delta N$, $\abs{V_i} = \abs{V_j}$ for all $i,j\geq 1$, and all except $\delta k^2$ of the pairs $V_i,V_j$ in the partition are $(\delta)$-regular.
\end{theorem}

We now state and prove the theorem for constructing a dense model for a sparse graph.
\begin{theorem}\label{thm:sparse-graph}
For every parameter $\gamma,\varepsilon,\varepsilon',\lambda'\lambda''>0$, such that $\lambda'\leq c\lambda \varepsilon'^2$  for a sufficiently large constant $c$. Then there exists $\eta\in[0,1]$ such that the following holds.
Let $\cB$ be a collection of graph-induced support access objects, such that for each $B^*\in\cB$, the graph it represents $G_{B^*}$ is  $(\eta,\gamma)$-upper-uniform. 

Let $\cD$ be a collection of distinguishers. If there exists an $(\varepsilon,\varepsilon',\delta ')$-auditor $\Lambda$ for the collection of sets $\sC = \set{(U_D,V_D)|D\in\cD}$,
then there exists an $(\varepsilon,\delta'')$-learning algorithm  $L$ for all $B^*\in\cB$ with respect to $\cD$.
\end{theorem}

\begin{proof}
At a high level, the proof of the theorem has two parts. In the first, \Cref{claim:exists-dense}, we use the regularity lemma for sparse graphs to show that for every upper-uniform graph $G$ there exists a dense graph $H$ that is indistinguishable from $G$. In the second part, we reduce finding a model to our graph $G^*$ to finding a model for a dense graph $H$ that is indistinguishable from it.

Let $\delta = \frac{\varepsilon'^2 }{100\gamma}$ . Let $(\eta',m) = \varphi(\delta,\gamma)$, where $\varphi$ is the function from the regularity lemma, \Cref{thm:sparse-graph}. Let $\eta = \min\set{\eta',1/(2m),\delta}$.

Let $B^*\in\cB$ be a graph-induced support-access object, representing a sparse graph $G^*$ that is $(\eta,\gamma)$-upper uniform. We remark that $G^*$ is also $(\eta',\gamma)$-upper uniform, as the upper-regularity is monotone in the $\eta$ parameter.

From  \Cref{claim:exists-dense}, there exists a graph $H$ with $\rho_H\approx 1/\gamma$ such that for all $U,V\subset [N]$, 
        \[ \abs{\frac{\abs{E_G(U,V)}}{\gamma} - \abs{E_H(U,V)}\rho_G }\leq 12\delta \rho_GN^2 .\]

Suppose that instead of a graph-induced support-access object $B^*$ representing the sparse graph $G^*$, we had access to an object $B^H$ representing the dense graph $H$. Then, we could have run the learning algorithm from  \Cref{sec:support-func} on $B^H$ (treating $B^H$ as a function-induced support-access object), and get an implementation of a model $M$ that is $2\varepsilon$-indistinguishable from $B^H$. Since the indistinguishability error is additive, $M$ is also $(2\varepsilon+\delta)$-indistinguishable from ${B^*}$.

In our case, we do not have access to the object $B^H$, and cannot find $H$. Instead, we simulate an auditor with access to $B^H$ by our auditor with access to $B^*$. Every time that the algorithm queries the auditor $\Lambda^{B^H,p}$, we instead invoke the auditor $\Lambda^{B^*,p}$ and forward its answer to the algorithm.

We claim that our simulated auditor is an $(2\varepsilon,\frac{1}{2}\varepsilon',\delta')$-auditor for $B^H$. In fact, a slightly weaker claim suffice. We claim that if $\Lambda^{B^*,p}$ returns a tuple $((U,V),b)$, then it is also a valid return value for $\Lambda^{B^H,p}$.

Suppose the auditor $\Lambda^{B^*,p}$ returned a tuple  $((U,V),b)$. Then with probability $1-\delta'$ we have 
\[ \abs{\Pr_{(u,v)\sim B^*}[u\in U,v\in V] - \Pr_{(u,v)\sim p}[u\in U,v\in V]} > \varepsilon'. \]
This is the same as 
\[ \abs{\frac{E_{G^*}(U,V)}{\rho_{G^*}\cdot  N^2} - \Pr_{(u,v)\sim p}[u\in U,v\in V]} > \varepsilon'. \]
From \Cref{claim:exists-dense}, we have that
\[ \abs{\frac{\abs{E_{G^*}(U,V)}}{\rho_{G^*} \cdot N^2} - \frac{\gamma\abs{E_H(U,V)}}{ N^2} }\leq 12\delta\gamma  .\]
Together with the fact that $\rho_H\in [1/\gamma-12\delta,1/\gamma+12\delta]$ we get that 
\[ \abs{\Pr_{(u,v)\sim B^H}[u\in U,v\in V] - \Pr_{(u,v)\sim p}[u\in U,v\in V]} \geq  \epsilon' - 12\delta\gamma \geq\frac{1}{2}\varepsilon'. \]
For our choice of $\delta$. That is, if the auditor $\Lambda^{B^*,p}$ returned a tuple $((U,V),b)$, then with probability $1-\delta'$, this answer is a valid return value also for $\Lambda^{B^H,p}$.

If the auditor $\Lambda^{B^*,p}$ does not return an answer, then $p$ is a $2\varepsilon$ indistinguishable model from $B^*$ (see the proof of \Cref{thm:rand-one} for more details), so we are done.
\end{proof}

\begin{claim}\label{claim:exists-dense}
    Let $\gamma,\delta,\eta >0$, be such that $\eta\leq \eta',1/(2m),\delta$ for $(\eta',m) = \varphi(\delta,\gamma)$. Then for every $(\eta,\gamma)$-upper uniform graph $G$ on vertex set $[N]$, there exists a graph $H$ on vertex set $[N]$ with $\rho_H\in [1/\gamma - 12\delta,1/\gamma+12\delta]$, such that for every $U,V\subset [N]$, we have that
    \[ \abs{\frac{\abs{E_G(U,V)}}{\gamma} - \abs{E_H(U,V)}\rho_G }\leq 12\delta \rho_GN^2 .\]
\end{claim}
\begin{proof}
    From the claim requirements, $\eta$ is a constant small enough that the regularity lemma for sparse graphs, \Cref{lem:regularity}, holds with $\gamma,\delta$. Therefore, any $(\eta,\gamma)$-upper uniform graph $G$ can be partitioned into 
    $V_0,V_1\ldots,V_k$ with $k\leq m$ such that $\abs{V_i}=\abs{V_j}$ for all $i,j\geq 1$, $\abs{V_0}\leq\delta N$ and all except $\delta k^2$ of the pairs $V_i,V_j$ are $(\delta)$-regular. Since $\eta \leq \frac{1}{2k}$ and $G$ is $(\eta,\gamma)$ regular, all sets $V_i,V_j$ satisfy $\rho_G(V_i,V_j)\leq\gamma\rho_G$.

    We use this partition to construct a dense graph $H$ as follows. For every pair $u,v\in [N]$, add the edge $(u,v)$ to $H$ with probability $p_{u,v}$ defined by:
    \[ p_{u,v} = \begin{cases}
        \frac{\rho_G(V_i,V_j)}{\gamma\rho_G} \quad &\text{if } u\in V_i,v\in V_j \text{ for a regular pair }V_i, V_j\\
        \frac{1}{\gamma} \quad &\text{otherwise}.
    \end{cases} \]
    From the fact that $\rho_G(V_i,V_j)\leq\gamma\rho_G$, we get that for every $u,v$, $p_{u,v}\leq 1$.

    We prove that $H$ satisfies the conditions of the claim. 
    Fix any two sets $U,V\subset [N]$.
    We start by bounding $\abs{E_G(U,V)}$,
    \[ \abs{E_G(U,V)} = \sum_{i,j}\abs{E_G(U\cap V_i, V\cap V_j)}. \]
    We bound this sum using the regularity lemma, and dividing into cases.
    \begin{enumerate}
        \item Exception set: From the regularity lemma, the size of the exception set if bounded by $\delta N$. Using the upper-uniformity condition, intersection with the exception set have:
    \[ \abs{E_G(V_0\cap V,U)}\leq \delta\gamma\rho_G N^2,  \]
    and the same holds for $E_G(V_0\cap U,V)$.
    \item Non-regular sets:  At most $\delta k^2$ of the sets $V_i,V_j$ in the partition are not regular. We bound the maximal number of edges between these sets:
    \[ \sum_{i,j \text{ non-regular}}\abs{E_G(V_i,V_j)}\leq \sum_{i,j \text{ non-regular}}\gamma\rho_G\cdot\abs{V_i}\abs{V_j}\leq \gamma\delta \rho_G N^2. \]
    \item Sets with small intersection: For the set $U$, we say that the intersection of $U$ with $V_i$ is small, if $\abs{U\cap V_i}\leq\delta\abs{V_i}$. Let $I_U\subset [k]$ be the parts that has small intersection with $U$, and $S_U = \cup_{i\in I_U}\paren{U\cap V_i}$ be their union. Then $\abs{S_U}\leq \delta N$. From the upper-uniformity condition, we have that 
    \[ \abs{E_G(S_U,V)}\leq \gamma\delta\rho_G N^2, \] and the same holds for $E_G(S_V,U)$.
    \item Regular sets with large intersection:  Let $V_i,V_j$ be a $(\delta)$-regular pair, such that $\abs{V_i\cap U}\geq\delta\abs{V_i}$ and  $\abs{V_j\cap V}\geq\delta\abs{V_j}$. Then according to the regularity lemma, \[ \abs{\rho_G(V_i\cap U,V_j\cap V) - \rho_G(V_i,V_j)}\leq \delta\rho_G .\]
    \end{enumerate}
    Combining all of the cases together, if we call $i,j\in [k]$ good if $V_i,V_j$ is regular, and $U,V$ have large intersection with $V_i,V_j$ we have that
       \begin{align*}
        \abs{{E_G(U,V)}} \leq {} & 5 \gamma\delta\rho_G N^2 + \sum_{\text{good }i,j}\abs{E_G(U\cap V_i,V\cap V_j)} \\\leq {} &  5 \gamma\delta\rho_G N^2  + \sum_{\text{good }i,j}\abs{U\cap V_i}\abs{V\cap V_j}\rho_G(V_i,V_j)\paren{1+\delta\rho_G}
    \end{align*}
   Similarly, we get a lower bound 
   \[ \abs{{E_G(U,V)}}\geq  \sum_{\text{good }i,j}\abs{U\cap V_i}\abs{V\cap V_j}\rho_G(V_i,V_j)(1- \delta\rho_G). \]
    
    We now bound $\abs{E_H(U,V)}$.
   The graph $H$ is generated such that each edge is sampled uniformly at random with probability $p_{u,v}$. For all good $i,j$, We have that if $u\in V_i,v\in V_j$, then $p_{u,v} = p_{i,j}=\frac{\rho_G(V_i,V_j)}{\rho_G\gamma}$. Since $\abs{V_i},\abs{V_j}= \Theta(N)$, we have from Chernoff bound that with high probability,
    \[ \abs{\rho_H(V_i\cap U,V_j\cap V) - p_{i,j}}\leq \delta p_{i,j} .\]
    
    There are at most $5 \delta N^2$ pairs of vertices $(u,v)$ that do not belong into a good pair of parts $V_i,V_j$.
    
    Together we get that
    \begin{align*}
        \abs{\abs{E_H(U,V)}\rho_G - \frac{\abs{E_G(U,V)}}{\gamma}} \leq {} & \sum_{\text{good }i,j}\abs{\abs{E_H(U\cap V_i,V\cap V_j)}\rho_G - \frac{\abs{E_G(U\cap V_i,V\cap V_j)}}{\gamma}}\\&+ 5\delta\rho_G N^2 + 5\delta\rho_G N^2 \\\leq {} & \sum_{\text{good }i,j}\abs{U\cap V_i}\abs{V\cap V_j}(\delta\rho_G(V_i,V_j) + \delta\rho_G) + 10\delta\rho_G N^2 \\ \leq {} &12\delta\rho_G N^2.
    \end{align*}

    The above equation also implies that $\rho_H\in [1/\gamma - 12\delta,1/\gamma+12\delta]$, by applying it with $U=V=[N]$.
\end{proof}

Next we show that a sparse graph with a dense subgraph does not have a dense indistinguishable model. This does not mean that it is not possible to learn a sparse graph with a dense subgraph, but it indicates that our approach have certain limits. 
This claim is not tight, in a sense that there are graphs do not satisfy the claims of \Cref{thm:sparse-graph}, and we do not know if they have a dense model.
\begin{claim}\label{claim:no-dense-graph}
Let $G$ be a graph, if there exists a distinguisher $D$ with vertex sets $U,V$ such that $\abs{U},\abs{V}\geq \varepsilon N$ and
\[  \frac{\rho_{G}(U,V)}{\rho_{G}} > \gamma + \frac{2\delta}{\varepsilon^2}, \]
then there is no graph $H$ with $\rho_H=1/\gamma$ that is $\delta$-indistinguishable from $G$ to $D$.
\end{claim}
\begin{proof}
If $H$ and $G$ are indistinguishable, it must be that 
\[ \abs{ \Pr_{(u,v)\sim G}[u\in U,v\in V]-  \Pr_{(u,v)\sim H}[u\in U,v\in V]}\leq\delta,\]
This is the same as
\[ \abs{\frac{\rho_{G}(U,V)}{\rho_{G}} - \frac{\rho_{H}(U,V)}{\rho_H}}\leq\delta \frac{N^2}{\abs{U}\abs{V}}.  \]
This means that 
\[ \frac{\rho_{H}(U,V)}{\rho_H}\geq \frac{\rho_{G}(U,V) }{\rho_{G}} - \frac{\delta}{\varepsilon^2} > \gamma, \]
which is a contradiction as $\rho_{H}(U,V)\leq 1$ and $\rho_H=1/\gamma$.
\end{proof}

\subsection{Learning Sparse Uniform Out-degree Graphs}\label{sec:const-out}
In this section we are interested in learning a sparse directed graph $G^*=(\{0,1\}^n,E^*)$, such that each $x\in G^*$ has a uniform constant out degree $d$. We assume $B^*$ is the object representing the graph $G^*$, and upon querying $B^*(\bot)$ outputs a random edge according to the following distribution: choose a random $x\in [N]$, then a random neighbor of $x$. 

If we denote the $d$ outgoing edges from each $x$ as $d$ functions, $f_1,\ldots,f_d:\{0,1\}^n\rightarrow\{0,1\}^n$, then learning $B^*$ is rather similar to learning these $d$ functions. Using this notations, a random edge from $B^*(\bot)$ corresponds to a pair $(x,f_i(x))$ for a random $x\in\{0,1\}^n,i\in[d]$.  If we would also assumes that the sample contains the index $i$, i.e. that $B^*(\bot)$ returns a random triplet $(x,i,f_i(x))$, then this would be exactly equivalent to learning the functions $f_1,\ldots,f_d$, and the result of \Cref{sec:bit-string-func} would translate automatically. Without this additional assumption we need to do a bit more work in order to translate the result of \Cref{sec:bit-string-func} into the graph object.

We formally describe the object, the class of distinguishers that we can work against and the auditor.

The object $B^*$ is a graph-induced support access object to a directed graph $G = (\set{0,1}^n,E)$ with out-degree $d$ for a constant $d$. The edge distribution is choosing a uniform $x\in \set{0,1}^n$ and then a random neighbor $y$ of $x$. We assume an order over the $d$ neighbors of each vertex, and denote by $f_1,\ldots,f_d:\set{0,1}^n\rightarrow\set{0,1}^n$ to be $d$ functions representing the neighbors of the vertices in the graph. 
\begin{description}
\item[Distinguishers:] Let $\cD$, such that each distinguisher $D\in \cD$ has an set $S_D\subset \set{0,1}^n$ and a coordinate $j\in[n]$. The distinguisher $D$ accept an edge $(u,v)$ if $u\in S_D$ and $v_j=1$.
\item[Auditor:] Let $\cD$ be defined as above. We say that an algorithm $\Lambda^{B^*,p}$ is an $(\varepsilon,\gamma,\delta)$ auditor for the collection of sets $\cS$ if it has the following properties. Given access to a graph-induced support access-object $B^*$ and query access to a predictor $p:\set{0,1}^n\rightarrow[0,1]$. If there exists $S\in\cS$ and $j\in[j]$ such that 
    \[
    b\paren{\Pr_{x,i}[x\in S,f_i(x)_j=1]-\E_{x}[ p_j(x)\cdot\one(x\in S)] } > \varepsilon. \]
    Then the auditor returns a set $S',j$ such that with probability $1-\delta$,
    \[
    b\paren{\Pr_{x,i}[x\in S',f_i(x)_j=1]-\E_{x}[ p_j(x)\cdot\one(x\in S')] }  > \gamma. \]
    The probability is uniform among all $x\in [N],i\in[d]$.
\end{description}

\begin{theorem}\label{thm:out-d}
Let $\cB$ be a collection of support access graph objects induced by $d$ functions $f_1,\ldots,\allowbreak f_d:\set{0,1}^n\rightarrow \set{0,1}^n$. Let $\cD$ be a collection of distinguishers as described above. 

Let $\varepsilon,\gamma,\delta',\delta''$ be parameters such that $\delta'\leq c\delta\gamma^2 n^{-1}$ for a sufficiently small constant $c$.
Let $\Lambda$ be an $(\varepsilon,\gamma,\delta')$ auditor $\Lambda$ for $\cD$. Then there exists a $(2\varepsilon,\delta)$-learner $L$ to $\cB$ with respect to the distinguisher class  $\cD$, and the learner $L$ runs in time $\poly(\gamma^{-1}\log(\delta^{-1})\alpha^{-1},W_1,W_2)$, where $W_1$ is the running time of the auditor $\Lambda$ and $W_2$ the circuit complexity of its output. The implementation $T$ that the learner outputs runs in time $\poly(\gamma^{-1}\log(\delta^{-1})\alpha^{-1},W_2)$
\end{theorem}

We prove the theorem in a similar way to the proof of \Cref{thm:func}. The learner $L$ and the implementation $T$ are the same as in the function, only we apply the same predictor $d$ times to get $d$ outgoing edges out of every vertex in the graph, and verify that the endpoints are distinct.
\begin{description}
\item[Implementation:] Given a list $F$ of tuples $(S,j,w)$, the implementation algorithm $T$ with oracle access to a random oracle $R$ generates a model $\hat B$. Given $x\in\set{0,1}^n$ the algorithm $T$ outputs $v$ by:
\begin{enumerate}
    \item For every $j\in[n]$ calculate $p_j\in [0,1]$ by: 
    \[p_j(x) = \lcap\paren{w\cdot\one(x\in S):(S,j,w)\in F}. \]
    \item For $t=1$ to $d$, choose a vertex $v^{(t)}$ by setting $v^{(t)}_j=1$ with probability $p_j(x)$, using the random oracle $R(v,t)$.\label{item:out-edge}
    \item If all of the vertices $v^{(t)}$ are distinct, choose a uniform $i\in[d]$ and return $v^{(i)}$. If there is a collision, go back to \Cref{item:out-edge} and repeat with new randomness from the random oracle.
    \end{enumerate}
We remark that since $d$ is constant and the domain $\{0,1\}^n$ is exponentially large, collisions are very unlikely. When analyzing the distribution of the model over large sets of vertices we can ignore this event.

\item[Learning Algorithm:] The learning algorithm $L$.
\begin{enumerate}
    \item Initialization: set $F = \set{(\set{0,1}^n,j,1/2)|j\in[n]}$.
    \item Query the auditor with the current model and $B^*$. If the auditor returns $(S,j,b)$, add $(S,j,b\cdot\gamma)$ to $F$ and repeat \label{item:graph-prod-update}.
\end{enumerate}
\end{description}
Notice that the learner is identical to the classic boosting algorithm, where we apply it individually on each of the coordinated of $\sum_{i\in[d]}f_i(x)$. This is the reason that we can only be indistinguishable with respect to a set of distinguishers $D$ that only check if a single output coordinate in $f(x)$ equals $1$. 

\begin{proof}[Proof of \Cref{thm:out-d}]
The proof is the same as the proof of \Cref{thm:func}, except the averaging over $i\in[d]$. The probability of collision on \Cref{item:out-edge} is negligible, so we can safely assume independence between the $d$ outgoing edges.

If the auditor does not output any set then the model is $\varepsilon$-indistinguishable for all $D\in\cD$.

We prove the correctness using a potential function, which sums over the error of each coordinate $j\in[n]$. 
Let \[\varphi(B^*,p_1,\ldots,p_n) =\sum_{j\in[n]}\sum_{x\in\set{0,1}^n}\paren{\sum_{i\in[d]}\frac{1}{d}\one(f_i(x)_j=1) - p_j(x)}^2. \]
Every $x$ appears in the sum $n\cdot d$ times, therefor the initial value of $\varphi$ is bounded by $n\cdot 2^n$.

Let $p_1,\ldots,p_n$ be the predictors after before the execution of \Cref{item:graph-prod-update} and $p'_1,\ldots,p_n'$ afterwards, then we claim that with probability $1-\delta'$ we have that
$\varphi(B^*,p_1,\ldots,p_n) - \varphi(B^*,p'_1,\ldots,p_n')\geq\gamma^2 2^n$.

The auditor returned $(S,i,b)$, then 
\begin{align}
    \abs{ \frac{1}{2^n}\sum_{x\in S}\E_{i\in[d]}[\one(f(x)_j=1)] - p_j(x)}\geq \gamma.\label{eq:graph-prod-adv}
\end{align}
Assume without loss of generality that $b=1$, then with probability $1-\delta$,
\begin{align*}
    & \varphi(B^*,p_1,\ldots,p_n) - \varphi(B^*,p'_1,\ldots,p_n') \\ = {} & \sum_{j\in[n]}\sum_{x\in S}\left(\paren{\E_{i\in[d]}[\one(f(x)_j=1)] - p_j(x)}^2 - \paren{\E_{i\in[d]}[\one(f(x)_j=1)] - p'_j(x)}^2\right)  \\
     = {} & \sum_{x \in S,j\in[n]}\left( -2 \E_{i\in[d]}[\one(f(x)_j=1)](p_j(x)- p'_j(x))+  p^2_j(x) - {p'_j(x)}^2\right).
\end{align*}
If there is not capping, we have
\begin{align*}
    \varphi(B^*,p_1,\ldots,p_n) - \varphi(B^*,p'_1,\ldots,p_n') = {} & \sum_{x\in S}\left(2\gamma\cdot\E_{i\in[d]}[\one(f(x)_j=1)] - 2\gamma p_i(x) -\gamma^2\right) \\
    \geq {} &
    2\gamma^2 2^n -\gamma^22^n\\
    = {} & \gamma^2 2^n,
\end{align*}
where we use \cref{eq:prod-adv} in the last inequality.

The values of $\E_{i\in[d]}[\one(f(x)_j=1)]$ are in $[0,1]$, therefore capping the value of $p_i$ to $[0,1]$ can only lower the value of the potential function.
\end{proof}

\subsection{Learning Uniform Degree Graphs}\label{sec:uniform}
Suppose we are interested in generating a truthful model for a uniform degree $d$ graph. That is, we want that all graph in our model has a uniform degree $d$. In \Cref{sec:dense}, we discussed applying the results from learning functions to learning graphs, and  have a truthful model for a directed graph with a uniform out-degree, assuming the graph is dense. In the previous section, we also saw how to have a sparse directed graph with uniform out-degree.

For undirected graphs, in \cite{goldreich2010} the authors showed an algorithm for generating a uniform degree graph that is indistinguishable from a random uniform degree graph. Their algorithm applies a random permutation on a large girth uniform degree expander graph.  

We are interested in generating a model $M$ that is indistinguishable from a specific graph-induced object $B^*$. Therefore, it is not possible to use the same approach. Instead, we are restricting the set of distinguishers to those that can be described by a partition over the set of vertices, and show how to create a uniform degree graph in this case. Our construction is simple - we sample and approximate the edge density between every two parts in the partition induced by the distinguishers, and the model is a random permutation over a deterministic graph with the correct edge density.

\paragraph{Distinguishers:}  Then the set of distinguishers $\cD$ contains distinguishers $D$ with sets $(U_D,V_D)$ such that $U,V\in\cU$. Every distinguisher $D$ accepts an edge $(u,v)$ if $u\in U,v\in V$. Let $\cU = \set{U\subset [N] | \exists D \st U=U_D \text{ or } U=V_D}$. We assume that $\cU$ is a partition with $t$ parts, and that $\abs {U_J}$ is linear in $N$.

\begin{lemma}\label{lem:uni-deg}
    Let $\cB$ be a collection of graph-induced support-access objects, such that for all $B^*\in\cB$, the graph $G_{B^*}$ has a uniform degree $d$. Let $\cD$ be the distinguishers class defined above.
    Then for every constant $\varepsilon$ there exists an $(\varepsilon,\delta)-$ learning algorithm $L$ for the class $\cB$ with respect to $\cD$. The algorithm runs in time $\poly(1/\varepsilon,\log(1/\delta))$.
\end{lemma}

We describe an implementation and a learning algorithm:
\begin{description}
\item[Representation:] The model $M$ is represented by the partition $\cU$, and by a function $k:\binom{t}{2}\rightarrow \mathbb{N}$. The function $k$ represents the number of edges between every two parts. We assume that it satisfies for all $j\in[t]$ we have $\sum_{i\in [t]}k(\set{i,j}) = d\cdot\abs{U_j}$.
\item[Implementation:] The implementation $T$ with a random oracle $R$ on partition $\cU$ and function $k$. It samples a random $u\in[N]$ it outputs an edge $(u,v)$ by:
\begin{enumerate}
\item Let $U_j\in\cU$ be the part in the partition that $u$ belongs to. We assume there is a deterministic order on the vertices in $U_j$, $u_1,\ldots,u_{\abs{U_J}}$.  
\item Choose a random $\ell\in[d]$. 
\item Let $E_{U_j}$ be a list of all the outgoing edges from $U_j$, by following order:
\[ E_{U_j} = (u^1_1,\ldots,u^1_{\abs{U_j}},u^2_1,\ldots,u_{\abs{U_j}}^d) \]
Let $F_{U_j} \in [t]^{\abs{E_{U_j}}}$ be the vector:
$F_{U_j} = (1^{k(\set{j,1})},2^{k(\set{j,2})},\ldots,t^{k(\set{j,t})})$. That is, $1$ appears exactly $k(j,1)$ times in $F_{U_j}$.
\item Calculate $i = F_{U_j}(u^{\ell})$. 
\item Use $R$ to apply a random matching between $\set{(u',l')\in U_j| F_{U_j}((u')^{l'})=i }$ to the set $\set{(u',l')\in U_i| F_{U_i}((u')^{l'})=j }$. Note that both sets have size $k(\set{i,j})$.
For $i=j$, pick an arbitrary matching that does not map any element to itself. If any vertex $u'$ appears in the set more than once, pick a matching from $U_j$ to the matching set in $U_i$, then for the second edge of $u'$ pick a shift of the matching (and vice verse).
\item Let $v\in U_i$ be the vertex matched to $(u,\ell)$ in the above random matching, output $v$.
\end{enumerate}
We remark that the random permutation is picked to be inevitable, such that if $v$ is picked, then the corresponding neighbor would be $u$.
\item[Learning algorithm:] We describe an algorithm for learning the function $k$.
\begin{enumerate}
    \item Sample edges from $B^*$, let $S=(u_1,v_1),\ldots,(u_m,v_m)$ be the sample. For every $i,j\in [t]$ set $k(\set{i,j})$ by for $i\neq j$
    \[k(\set{i,j}) = \frac{\abs{\set{(u,v)\in S|u\in U_i,v\in U_j}}}{\abs{S}}d N. \]
    and for $i=j$ twice this value.
    \item For every $i\in[t]$, calculate  $e_i  = \sum_{j\in[t]}k(\set{i,j})-d\abs{U_i}$.
    \item If there exists $i\in[t]$ such that $\abs{e_i} \geq d\varepsilon/2\abs{U_i}$, quit the algorithm and output an arbitrary $k$ such that $e_i=0$ for all $i$.
    \item While there exists $i$ with $e_i\neq 0$: Pick $i$ with $e_i>0$ and $i
    '$ with $e_{i'}<0$. Pick a random $j\in[t]$, such that it is possible to decrease $k(\set{i,j})$ by $1$ and increase $k(\set{{i'},j})$ by $1$, and do so.
    \label{item:fix-rest}
\end{enumerate}
\end{description}
\begin{remark}
The initial value $k(\set{i,j})$ on the first item are created such that  $\sum_{i,j\in[t]}k(i,j)=2dN$, and therefore we have $\sum_{i\in [t]}e_i=0$. This means that in \Cref{item:fix-rest}, if there is $i$ with $e_i>0$ then there must also be $i'$ with $e_{i'}<0$. Therefore, the correction algorithm always ends.
\end{remark}

The implementation $T$ outputs exactly $d$ neighbors for each vertex $u\in U$. We prove that the resulting model is indistinguishabile for every $D\in\cD$.

\begin{proof}[Proof of \Cref{lem:uni-deg}]
We start by proving that with high probability over the samples of the learning algorithm, we have $\abs{e_i}<\frac{\varepsilon}{2}\abs{U_i}$ for all $i\in [N]$.

Fix some $i\in [t]$. The specification graph $B^*$ has uniform degree $d$, which implies that
\[ \sum_{j\in [t]}\abs{\set{(u,v)\in B|u\in U_i,v\in U_j}}=d\abs{U_i}. \]

For every edge $(u,v)$ sampled, let $I_{(u,v)}$ be the indicator random variable that this edge contributes to $k(\set{i,j})$ for some $j\in[t]$. Then we know that the expected value of $\sum_{(u,v)\in S}I_{(u,v)} = \frac{\abs{U_i}}{N}\cdot \abs{S}$.
Therefore, by chernoff bound, with probability $\exp(\frac{\abs{U_i}}{N}\cdot \abs{S}\varepsilon^2)$, all approximations are within $\varepsilon/2$ form the expected value.

The corrections on the value of $k$ on \Cref{item:fix-rest} are negligible, as they correct at most $t^2 d$ edges per group $U_j$ of size linear in $N$. 

Let $D$ be a distinguisher with sets $U_i,U_j$ in the partition. Then by the model, the probability that an edge is sampled with $U_i,U_j$ is  $k(\set{i,j})/(dN)$. Therefore, if $k_{i,j}$ are correct up to an $\varepsilon/2 Nd$ factor, the model is indistinguishable for $D$.  
\end{proof}

\section{Impossibilities}
\label{sec:impossibility}

A main difference in our work from \cite{goldreich2010} is in the target distribution/object we aim to be indistinguishable from. In \cite{goldreich2010}, the target distribution is fixed and uniform over many objects, whereas in our setup the target is a single object which is initially unknown, and a learner is needed to access the target object to make it possible to create an indistinguishable model.
This difference makes our setup challenging, and below we show example tasks that are impossible to achieve in our setup because of this difference.

\paragraph{Fooling Distinguishers with Entry-Access is Hard.}
In \cite{goldreich2010}, the distinguishers can query for specific entries of an object. Such distinguishers can be impossible to fool in our setup. For example, suppose the target object $B^*$ is the \emph{entry-access} object induced by a function $f^*:X\to\{0,1\}$ (\Cref{def:entry-function}), and suppose our learner aims to output a model $M$ of entry-access objects $B$ induced by functions $f:X\to\{0,1\}$. For every $x\in X$, suppose there is a distinguisher that queries for the value of $f(x)$ and outputs $\acc$ if and only if $f(x) = 1$. To fool these distinguishers, we have to learn the target function $f^*$ exactly, which is clearly impossible if the domain $X$ has exponential size and the learner can only make polynomially many queries.
\begin{theorem}
\label{thm:entry-hard}
Let $X$ be a non-empty finite set. Let $\cB$ be the class of entry-access objects induced by all functions $f:X\to \{0,1\}$. Let $\cD$ be the class of distinguishers $D_{x}$ for every $x\in X$ where given an object $B$, the distinguisher $D_{x}$ outputs $\acc$ if and only if the answer $a\sim B(x)$ is equal to $1$. Let $L$ be an $(\varepsilon,\delta)$-learner for the class $\cB$ w.r.t.\ $\cD$ for $\varepsilon,\delta < 1/2$. Then $L$ needs to query every input $x\in X$ in the worst case.
\end{theorem}
\begin{proof}
For the sake of contradiction, let us assume that whenever $L$ outputs a model, there always exists $x\in X$ that is not queried by $L$.
Consider the case where the target object $B^*\in \cB$ is drawn uniformly at random. That is, $B^*$ is induced by $f^*:X\to \{0,1\}$ chosen randomly such that for every $x\in X$, $f^*(x)$ is distributed independently and uniformly from $\{0,1\}$. When $L$ outputs a model $M$ given entry access to the random target model $B^*$, the conditional distribution of $f^*(x)$ for the $x\in X$ not queried by $L$ is still the uniform distribution over $\{0,1\}$, and thus with conditional probability at least $1/2$ over the randomness in $f^*$, it holds that $\Pr_{B\sim M, a\sim B(x)}[a = f^*(x)] \le 1/2$, and thus
\[
|\Pr[D_x^{B^*} = \acc] - \E_{B\sim M}[\Pr[D_x^B = \acc]]|\ge 1/2.
\]
By the law of total probability, with (unconditional) probability at least $1/2$, there exists a distingsher $D\in \cD$ such that
\[
|\Pr[D^{B^*} = \acc] - \E_{B\sim M}[\Pr[D^B = \acc]]|\ge 1/2.
\]
This contradicts the fact that $L$ is an $(\varepsilon,\delta)$-learner with $\varepsilon,\delta < 1/2$.
\end{proof}

\paragraph{Learned Model Needs to be Stronger than Distinguishers.}
The model learned in \cite{goldreich2010} can fool distinguishers with significantly larger circuit complexity than the model itself. Below we show that this can become impossible in our setup where the target is a single object. 
\begin{theorem}[Remark 1.6 in \cite{TrevisanTV09}]
\label{thm:strong-model}
Let $n,W > 1$ be positive integers satisfying $W\log W \le 2^n/C$ for a sufficiently large absolute constant $C > 0$.
There exists a sample-access object $B^*$ induced by a function $f^*:\{0,1\}^n \to \{0,1\}$ and a distinguisher $D$ with circuit complexity $\tilde O(nW)$ such that for any model $M$ with circuit complexity at most $W$, it holds that
\begin{equation}
\label{eq:complex-learner-0}
|\Pr[D^{B^*} = \acc] - \E_{B\sim M}[\Pr[D^B = \acc]]| > 1/3.
\end{equation}
\end{theorem}
\begin{proof}
Define $k := CW\log W \le 2^n$.
By known constructions of $k$-wise independent hash functions \cite{MR532173,CarterWe81}, there exists a distribution over $f^*:\{0,1\}^n\to \{0,1\}$ with the following properties:
\begin{enumerate}
\item for any distinct $x_1,\ldots,x_k\in \{0,1\}^n$, the $k$-tuple $(f^*(x_1),\ldots,f^*(x_k))$ is distributed uniformly from $\{0,1\}^k$;
\item every $f^*$ drawn from the distribution has circuit complexity $\tilde O(nk) = \tilde O(nW)$.
\end{enumerate}
Consider any fixed model $M$. By concentration inequalities for $k$-wise independent random variables \cite{365687}, it holds that 
\[
\Pr_{f^*}\Big[\Pr_{B\sim M, (x,y)\sim B(\bot)}[y = f^*(x)]\ge  2/3 \Big ] < 2^{-\Omega(k)} = 2^{-\Omega(CW\log W)}.
\]
Choosing $C$ to be sufficiently large and applying the union bound over all $2^{O(W\log W)}$ models of circuit size at most $S$, there exists $f^*$ with circuit complexity $\tilde O(nW)$ such that for every model $M$ with circuit size at most $W$,
\begin{equation}
\label{eq:complex-learner-1}
\Pr_{B\sim M, (x,y)\sim B(\bot)}[y = f^*(x)]<  2/3.
\end{equation}
Now we choose $D$ to be the distinguisher that on a sample $(x,y)$ outputs $\acc$ if and only if $y = f^*(x)$. Inequality \eqref{eq:complex-learner-1} implies
\begin{equation}
\label{eq:complex-learner-2}
\E_{B\sim M}[\Pr[D^B = \acc]] < 2/3.
\end{equation}
Let $B^*$ be the object induced by $f^*$. By our definition of $D$ and $B^*$, we have $\Pr[D^{B^*} = \acc] = 1$. Combining this with \eqref{eq:complex-learner-2}, we get \eqref{eq:complex-learner-0} as desired.
\end{proof}
\paragraph{The Distinguisher Class Needs to be Learnable.}
Since the target distribution in \cite{goldreich2010} is fixed, no learning is needed in order to produce an indistinguishable model. In our setup, the learning task is usually performed using an auditor, which can be viewed as a weak agnostic learner for the class of distinguishers. A natural question is whether we can still achieve indistinguishability if such a weak agnostic learner does not exist. 
Previous works \cite{hebert2018multicalibration,gopalan_et_al:LIPIcs.ITCS.2023.60} have shown negative answers to this question for certain notions of indistinguishability (such as calibrated multiaccuracy) by showing that these notions imply (strong) agnostic learning for the distinguisher class. The indistinguishability notion we use for generative models is closer to multiaccuracy, and below we show that efficiently achieving this notion requires the distinguisher class to be efficiently realizably learnable. For a true function $f^*:X\to \{0,1\}$, multiaccuracy requires a predictor $p:X\to [0,1]$ to satisfy
\begin{equation}
\label{eq:impossibility-ma}
|\E[(f^*(x) - p(x))g(x)]| \le \varepsilon
\end{equation}
for every function $g$ in a class $\cG$. Now consider the case where $\cG$ consists of functions $g:X\to \{-1,1\}$. For an arbitrary $g^*\in \cG$, suppose the true function $f^*$ satisfies $f^*(x) = 1$ if $g^*(x) = 1$ and $f^*(x) = 0$ if $g^*(x) = -1$. Then \eqref{eq:impossibility-ma} implies
\begin{equation}
\label{eq:impossibility-ma-1}
\E|f^*(x) - p(x)| \le \varepsilon.
\end{equation}
Now we define $\hat g(x) = 1$ if $p(x) \ge 1/2$, and define $\hat g(x) = -1$ if $p(x) < 1/2$. It is easy to check that if $\hat g(x)\ne g^*(x)$ for some $x\in X$, then $|f^*(x) - p(x)| \ge 1/2$, and thus \eqref{eq:impossibility-ma-1} implies the following realizable learning guarantee for the class $\cG$:
\[
\Pr[\hat g(x)\ne g^*(x)] \le 2\varepsilon.
\]

\bibliographystyle{alpha}
\bibliography{ref}
\end{document}